\documentclass[sn-mathphys]{sn-jnl}

\jyear{2022}%

\theoremstyle{thmstyleone}%
\newtheorem{theorem}{Theorem}
\newtheorem{proposition}[theorem]{Proposition}%

\usepackage{amsmath}
\usepackage{amssymb}
\usepackage{amsfonts}
\usepackage{amsthm}
\usepackage{graphicx}
\usepackage{physics}
\usepackage{subcaption}

\newcommand{\ds}{\displaystyle}

\newcommand{\mbf}{\mathbf}
\newcommand{\jac}{J_{\mbf u} \,}
\newcommand{\bs}{\boldsymbol}

\newcommand{\evs}[1]{\textcolor{black}{#1}}

\newtheorem{lemma}{Lemma}
\newtheorem{corollary}{Corollary}

\DeclareMathOperator{\diag}{diag}

\theoremstyle{thmstyletwo}%
\newtheorem{remark}{Remark}%

\theoremstyle{thmstylethree}%
\newtheorem{definition}{Definition}%

\begin{document}

\title[General conditions for Turing and wave instabilities in reaction-diffusion systems]{General conditions for Turing and wave instabilities in reaction-diffusion systems}


\author[1]{\fnm{Edgardo} \sur{Villar-Sep\'ulveda}}\email{edgardo.villar-sepulveda@bristol.ac.uk}

\author[1]{\fnm{Alan R.} \sur{Champneys}}\email{a.r.champneys@bristol.ac.uk}


\affil[1]{\orgdiv{Engineering Mathematics}, \orgname{University of Bristol}, \orgaddress{\street{Ada Lovelace Building, Tankard's Cl, University Walk}, \city{Bristol}, \postcode{BS8 1TW}, \state{Somerset}, \country{United Kingdom}}}




\abstract{Necessary and sufficient conditions are provided for a diffusion-driven instability of a stable equilibrium of a reaction-diffusion system with $n$ components and a diagonal diffusion matrix. These can be either Turing or wave instabilities. Known necessary and sufficient conditions are reproduced for there to exist diffusion rates that cause a Turing bifurcation of a stable homogeneous state in the absence of diffusion. The method of proof here though, which is based on a study of dispersion relations in the contrasting limits in which the wavenumber tends to zero and to $\infty$, gives a constructive method for choosing diffusion constants. The results are illustrated on a model for the dispersion of malaria, a 3-component FitzHugh-Nagumo-like model proposed to study excitable wavetrains, and for two different coupled Brusselator systems with 4 components.}




\maketitle

\section{Introduction}\label{sec1}
    The formation of patterns and structures in a chemically reacting medium due to differential rates of diffusion of an otherwise stable equilibrium was first proposed in the last published work of Alan Turing \cite{Turing1952}. It seems Turing was on the verge of explaining how this instability underlies the formation of structures in plants such as the daisy; see \cite{dawturing}. The recent issue of {\em Philosophical Transactions} \cite{PhilTrans} highlights the state of the art in how the patterns that emerge from a Turing instability have been used to explain many phenomena in biology. There is also a parallel literature in the wider theory of pattern formation across many different length scales of physics from micromechanics to spatial ecology; see also the other recent volumes \cite{Meron,Woodsissue} and references therein.
    
    Nevertheless, although the conditions for having a Turing bifurcation have been generalized to systems with $n$ components, sometimes it remains unclear whether a system has the opportunity to develop these spatial patterns. The most definitive results so far are due to Satnoianu, Menzinger, and Maini \cite{Sat}. They used the definition of a so-called s-stable Jacobian matrix to determine necessary and sufficient conditions for a Turing bifurcation in a wide class of reaction-diffusion systems with an arbitrary number of components. Their approach, which relies on the Routh-Hurwitz criterion, does not provide precise details on the form of the dispersion relations upon adding diffusion, nor definitive conditions on which parameters might lead to a bifurcation.
    
    We note also the work \cite{Kilka} that extended some of the results in \cite{Sat} to systems that have components with zero diffusion coefficients, and \cite{Hoang2013} that shows that linear instability implies a nonlinear bifurcation under certain technical assumptions.  Finally, we mention the work of Kuznetsov and Polezhaev \cite{kuznetsov2020inspiration} who find new ways of obtaining Turing bifurcations by considering what happens to dispersion relations in the limit that some diffusion coefficients tend to zero. Their results, which are mostly restricted to 3-component systems, provide a key motivation for the method introduced in this paper.
    
    In all of these works, though, it is implied that the critical eigenvalue at the Turing bifurcation is real. The case of a complex critical eigenvalue, leading to a so-called wave bifurcation (or finite wavenumber Hopf bifurcation) has received less attention. \evs{A number of authors, e.g.~\cite{Korvasova} have shown that such an instability requires at least a three-species model.} A notable work is that by Anma, Sakamoto, and Yoneda \cite{Anma} that gives sufficient conditions for such a bifurcation for 3-component systems. See section \ref{sec:Hopf} below for a recapitulation of their result.
    
    Given the recent focus on synthetic biology, several authors have sought to understand the design principles of interacting chemical components necessary to achieve Turing instability, see e.g.~\cite{Sharpe,Shi}. Scholes {\em et al} \cite{Scholes} have performed a comprehensive search of the linear properties of all possible network topologies of two or three interacting chemical species to see how common and robust Turing instability is.  They sample both the functional forms of interaction and the values of diffusion constants and other parameter values. Their conclusion is Turing patterns are common but not robust in some sense. Note that they do not consider conditions for bifurcations (e.g.~upon increasing diffusion rates) nor the possibility of wave bifurcations.
    
    A similar approach is taken by Haas and Goldstein \cite{Haas}, except they use a sample of values of the Jacobian matrices of linearised kinetics in reaction-diffusion systems with $n$ components up to $n=6$. They argue that the minimum diffusion threshold required for instability decreases with $n$.
    
    In the present paper, we develop a general approach by studying dispersion relations of $n$-component reaction-diffusion systems linearised \evs{around} a homogeneous steady state.  \evs{We suppose that a researcher is studying a reaction-diffusion system and asks the question whether, for the given kinetic terms, there are values of the diffusion coefficients that lead to either a Turing or a wave bifurcation. In order to answer this question, we take various limits, which strictly speaking are likely to go outside the bounds for which the model is strictly valid, for example, that certain diffusion coefficients vanish or are sufficiently large. This leads to conditions on principal sub-matrices of the homogeneous system, for which we can make precise statements about the existence of either kind of bifurcation for a finite range of values of the diffusion rates.}
    
    In what follows, we consider a system of $n$ reaction-diffusion equations in time, $t$, and space $\mbf x \in \mathbb R^m$, where each state variable $u_i(\mbf x, t)$ $i=1,\ldots,n$ satisfies the system of partial differential equations (PDEs)
    \begin{align}
    	\mbf u_t&=\mbf f(\mbf u) + \mathbb D \, \nabla^2 \mbf u. \label{generalsystem}
    \end{align}
    Here, $\mbf u=\left(u_1,u_2,\ldots ,u_n\right)^\intercal$ is a vector of variables depending on $t$ and $\mbf x$,
    \begin{align*}
    	\mathbb D&=\diag\left(D_1,D_2,\ldots ,D_n\right)
    \end{align*}
    where $D_i\geq 0$ is a diffusion constant, and $\mbf f=\left(f_1,f_2,\ldots ,f_n\right)^\intercal$ is a vector field assumed to be sufficiently smooth (at least of class $\mathcal C^1$).
    \begin{remark} 
    	\begin{enumerate}
            \item For a well-posed partial differential equation, \evs{it would be necessary to specify a domain $\Omega \subset \mathbb{R}^n$, subject to suitable boundary conditions. However, in order to be general, we shall use the approach commonly taken in reaction-diffusion systems, of considering the limit $\Omega \to \mathbb{R}^n$ such that solutions remain bounded at infinity. Then, it is well known that the discrete spectrum of $\nabla^2$ subject to, say, Neumann boundary conditions, tends to a continuous spectrum. In what follows, for definiteness, we shall assume that all domains constructed in this way are the limits of hypercubes so that the spectrum we are interested in is spanned by eigenfunctions of the form $u_i= \cos(k_ix_i)$, $k_i \geq 0$ of the operator $\nabla^2$ corresponding to the negative eigenvalue $- k^2 = - \abs{\mbf k}^2$ where $\mbf{k} = \left(k_1, k_2, \ldots, k_m\right)^\intercal$ is a vector of {\em wavenumbers}. This enables the derivation of dispersion relations of temporal eigenvalues $\lambda_i(k)$; see for example \cite{Murray2001}.} On a finite domain with typical boundary conditions, the dispersion relation becomes discretised, and we expect to see a sequence of symmetry-breaking bifurcations that accumulate as the length of the domain increases at the parameter values at which an infinite-domain Turing bifurcation occurs (see e.g. \cite{Victor3}).
    		\item The form of \eqref{generalsystem} is chosen for convenience. Under suitable additional hypotheses, the results in this paper can be easily extended to a broader class of PDEs: systems with non-diagonal diffusion matrices, different temporal timescales, and more general spatial operators. See section \ref{sec:discuss} for a discussion.
    	\end{enumerate}
    \end{remark}
    Let $\mbf P$ be a homogeneous steady state of \eqref{generalsystem}. That is, $\mbf{f}(\mbf{P}) = 0$. We will denote the Jacobian matrix of this system at $\mbf P$ by	
    \begin{align*}
    	\jac \mbf f(\mbf P) = \left(\frac{\partial f_i}{\partial u_j}(\mbf P)\right)_{1 \leq i, j \leq n}.
    \end{align*}
    This implies that the linear part of the full system at $\mbf P$ is given by $\jac \mbf f(\mbf P) + \mathbb D \, \nabla^2$. 
    
    Next, as we want to analyze Turing and wave instabilities, we will be interested in the space of eigenfunctions of the operator $\nabla^2$ with negative eigenvalues, i.e., vector functions $\mbf u$ that satisfy the following equation:
    \begin{align*}
    	\nabla^2 \mbf u + \abs{\mbf k}^2 \, \mbf u = \mbf 0,
    \end{align*}
    where $\mbf k = \left(k_1, k_2, \ldots, k_m\right)^\intercal$.  This implies that the Jacobian matrix of our system at $\mbf P$, restricted to this space, is given by $\jac \mbf f(\mbf P) - \abs{\mbf k}^2 \, \mathbb D$.
    
    In particular, considering $\abs{\mbf k} = k$ as a real number, a {\em diffusion-driven instability} arises if $\mbf P$ is stable in the absence of diffusion, yet becomes unstable in the presence of a specific real wavenumber $k$ and a diffusion matrix $\mathbb{D}$. In particular, we will assume that all eigenvalues of $\jac \mbf f(\mbf P)$ have negative real part, then we shall seek a diffusion matrix $\mathbb D$ and wavenumber $k\in \mathbb R$ such that $\jac \mbf f(\mbf P)-k^2 \, \mathbb D$ has eigenvalues in the right-half of the complex plane. In particular, let $\lambda(k)$ be an eigenvalue of $\jac \mbf f(\mbf P)-k^2 \, \mathbb D$  for which  $\mbox{Re}(\lambda(k)) >0$ for some $k>0$. If we allow the diffusion rates (entries of $\mathbb{D}$) to act as continuous parameters, then we are interested in the transition depicted in Fig.~\ref{fig:dispersionrelations}. In panel (a), we have no diffusion-driven instability since the system is stable for every value of  $k$. Panel (b) depicts \evs{a {\em bifurcation point}, which we define to occur at a critical value of a parameter (either diffusion coefficient or otherwise) at which there is a quadratic tangency of the dispersion curve $\mbox{Re}(\lambda(k))$ at $\mbox{Re}(\lambda(k))=0$.} \evs{This} could be a Turing bifurcation or a wave one depending on whether $\lambda\left(k^*\right)$ is real or purely imaginary, respectively. Finally, panel (c) \evs{depicts} that the instability patterns that exist correspond to perturbations to $\mbf P$ that have wavenumbers $k_1<k<k_2$.
    \begin{figure}
    	\centering
    	\begin{subfigure}[b]{0.32\textwidth}
    		\centering
    		\includegraphics[width=\textwidth]{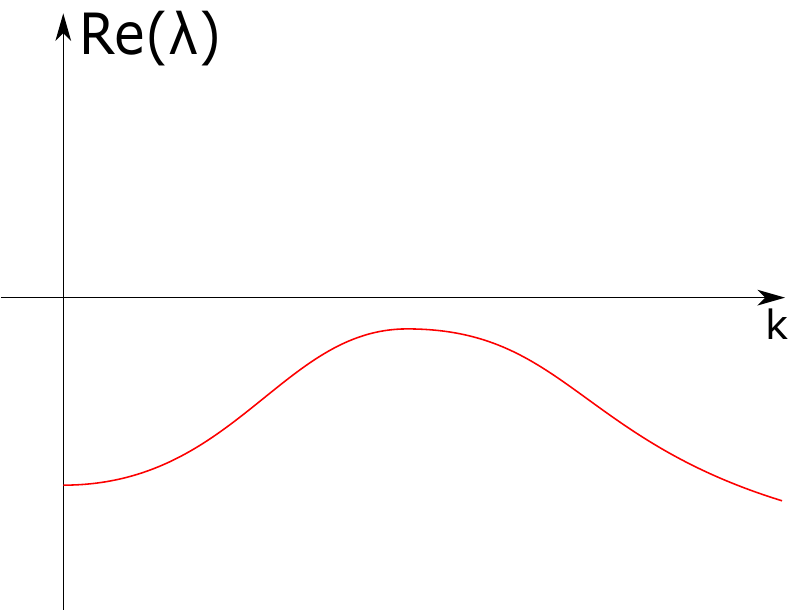}
    		\caption{ }
    	\end{subfigure}
    	\hfill
    	\begin{subfigure}[b]{0.32\textwidth}
    		\centering
    		\includegraphics[width=\textwidth]{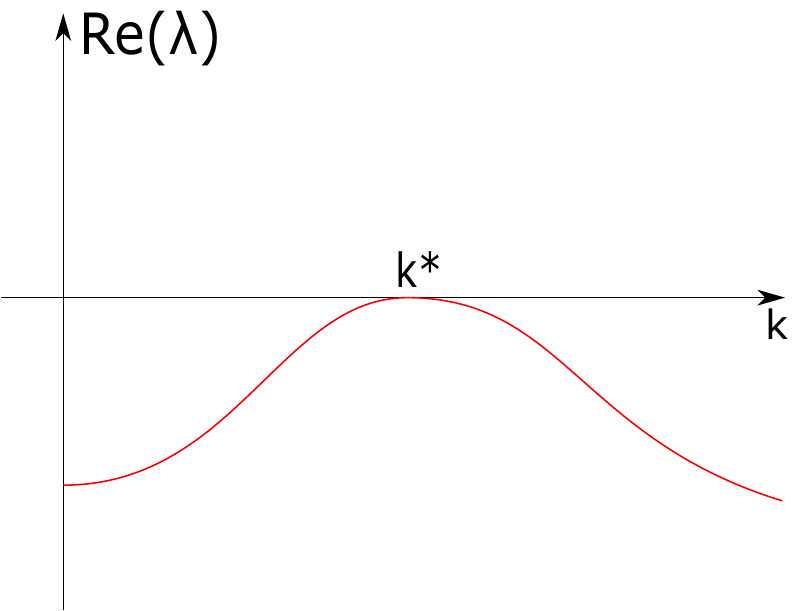} \caption{ }
    	\end{subfigure}
    	\hfill
    	\begin{subfigure}[b]{0.32\textwidth}
    		\centering
    		\includegraphics[width=\textwidth]{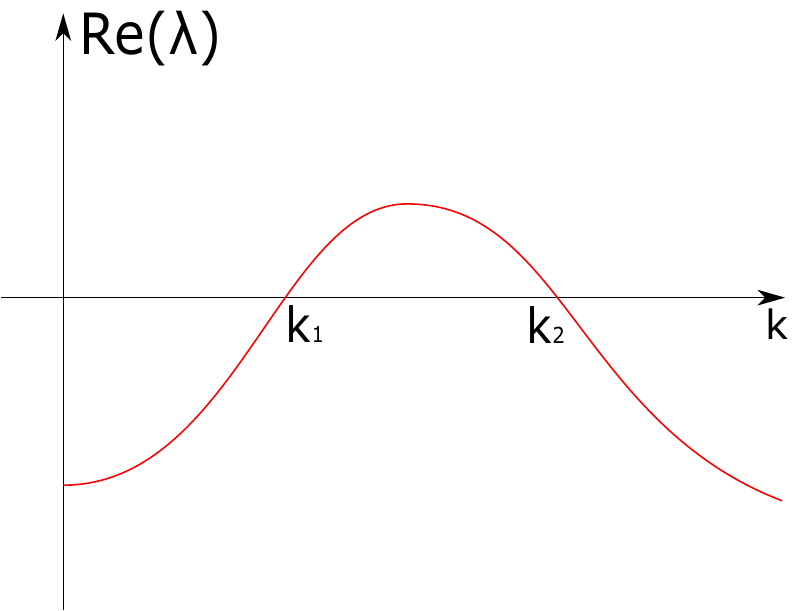}
    		\caption{ }
    	\end{subfigure}
    	\caption{Sketches of {\em dispersion relations}, that is the locus of the real part of an eigenvalue $\lambda(k)$, for different parameter values. Upon parameter variation we have: (a) absence of instability; (b) a bifurcation point giving onset to instability; (c) unstable patterns exist for wavenumbers $k_1<k<k_2$. For simplicity we depict only the dispersion curve of a mode for which $\max_k(\mbox{Re}(\lambda_i(k)))$ is largest.}
    	\label{fig:dispersionrelations}
    \end{figure}
    
    The rest of this paper is outlined as follows. In \evs{Section} \ref{sec:prelim}, we give some preliminary results that are used to prove the main theorems of the paper, ending with a simple proof of a simple part of the main result in \cite{Sat}, which will be used later. Section \ref{sec:Hopf} considers conditions for wave instabilities, confirming that the minimal number of components for such instabilities to occur is $n = 3$, and gives sufficient conditions for this case. Section \ref{sec:general} then provides conditions for both Turing and wave instabilities to occur in general $n$-component models. Section \ref{sec:examples} contains examples that illustrate the theory developed throughout the paper and \evs{Section} \ref{sec:discuss} contains concluding remarks.

\section{Preliminary results}\label{sec:prelim}
    In this section, we introduce some general results that will form the core methodology of this work, see specifically \evs{lemmas} \ref{positiveeig} and \ref{nonpositiveeig} below. We build up these results via some elementary \evs{lemmas} that are largely already known from literature on matrix stability, see e.g.~ \cite{Cross,Sat,Grantmacher} but will prove useful in establishing the setting and notation of our theory.
    
    We begin with the following elementary result which will apply to characteristic polynomials of matrices. 
    \begin{lemma} \label{positivitylemma}
    	Let $\lambda_1, \lambda_2, \ldots, \lambda_m\in \mathbb C\setminus \{0\}$ be $m$ non-zero complex numbers, which are roots of a polynomial equation with real coefficients. Then,
    	\begin{align}
    		(-1)^m \prod_{j = 1}^m \lambda_j<0 \label{negativeproduct}
    	\end{align}
    	if and only if there exists an odd number of indices $q\in \{1,2,\ldots ,m\}$ such that $\lambda_q \in \mathbb R$ and $\lambda_q>0$.
    \end{lemma}
    
    \begin{proof}
        \evs{First, as $\lambda_1, \lambda_2, \ldots, \lambda_m$ are roots of a polynomial equation with real coefficients, then all complex non-real numbers come in complex-conjugate pairs. With this, their product can be decomposed as}
        \begin{align}
    		\prod_{j=1}^m \lambda_j = \prod_{i=1}^{2p} \lambda_i \prod_{\ell=2p+1}^r \lambda_\ell \prod_{s=r+1}^m \lambda_s, \label{eigvalprod}
    	\end{align}
    	\evs{where the first product on the right-hand side of \eqref{eigvalprod} has all the complex non-real eigenvalues, whilst the second and third ones are the real positive and negative eigenvalues, respectively.}
    	
    	\evs{The first two products are positive, so the sign of the whole expression is given by the sign of the third term, which depends on the parity of positive roots, as stated.}
    \end{proof}
    
    Next, we shall introduce a notation for the {\em principal submatrices} and {\em principal minors} of a square matrix $A$. Note that the notation used in \cite{Cross} uses the complement of the set of indices used here.
    \begin{definition}\label{def:minors}
    	For each integer $1\leq k\leq n$, we define a {\bf word} to be a $k$-tuple of non-repeating integers $\left(i_1,i_2,\ldots ,i_k\right)$ with $1\leq i_j\leq n$ for each $j=1,2\ldots ,k$, and we shall call it an {\bf increasing word} if $i_1<i_2<\ldots <i_k$. For each increasing word, we refer to the corresponding {\bf principal submatrix} $S_{i_1 i_2 \ldots i_k}$ of $\jac \mbf f(\mbf P)$ as the $k\times k$ matrix that takes into account only its rows and columns with indices $i_1, i_2, \ldots, i_k$. We also denote by $M_{i_1 i_2 \ldots i_k}$ the {\bf principal minor} of $\jac \mbf f(\mbf P)$ which is the determinant of $S_{i_1 i_2 \ldots i_k}$. 
    \end{definition}
    With this, note that Lemma \ref{positivitylemma} can be used directly to prove the following:
    \begin{lemma} \label{wonderfulemma}
    	There exist an integer $1\leq k\leq n$ and an increasing word $\left(i_1,i_2\ldots i_k\right)$ such that $(-1)^k M_{i_1i_2\ldots i_k}<0$ if and only if $S_{i_1i_2\ldots i_k}$ has an odd number of positive real eigenvalues.
    \end{lemma}
    \begin{proof}
    	Note that the condition $(-1)^k M_{i_1i_2\ldots i_k}<0$ is equivalent to saying that $(-1)^k$ times the $k$ eigenvalues of $S_{i_1i_2\ldots i_k}$ is negative. Hence, owing to Lemma \ref{positivitylemma}, this is true if and only if there are an odd number of positive real eigenvalues of $S_{i_1 i_2 \ldots i_k}$.
    \end{proof}
    Next, we state a result that will be implicitly assumed in what follows under the requirement that $\mbf P$ is a stable equilibrium point in the absence of diffusion.
    \begin{lemma} \label{necessarystability}
    	Let $\mbf P$ be a homogeneous steady state of \eqref{generalsystem} in the absence of diffusion. If $\mbf P$ is stable, then
    	\begin{align*}
    		(-1)^k\sum_{1\leq i_1<i_2<\ldots <i_k\leq n} M_{i_1i_2\ldots i_k}> 0,
    	\end{align*}
    	for all $k=1,2,\ldots ,n$.
    \end{lemma}
    \begin{proof}
    	Based on a standard result on the determinant of sums of matrices \cite{Marcus1990}, we know that for any pair of complex-valued $n\times n$ matrices, the determinant of $A+B$ can be written as
    	\begin{align}
    		\det(A + B) = \sum_{r = 0}^n \sum_{\bs \alpha, \bs \beta} (-1)^{s(\bs \alpha) + s(\bs \beta)} \det(A[\bs \alpha \vert \bs \beta]) \, \det(B(\bs \alpha \vert \bs \beta)), \label{determinantofsum}
    	\end{align}
    	where the inner sum is over all strictly increasing words $\bs \alpha, \bs \beta$ of length $r$; $A[\bs \alpha \vert \bs \beta]$ is the $r\times r$-square submatrix of $A$ formed of the rows in $\bs \alpha$ and columns in $\bs \beta$; $B(\bs \alpha \vert \bs \beta)$ is the $(n - r)$-square submatrix of $B$ formed of the rows complementary to $\bs \alpha$ and columns complementary to $\bs \beta$; and $s(\bs \alpha)$ (respectively, $s(\bs \beta)$) is the sum of the integers in $\bs \alpha$ (respectively, $\bs \beta$). In particular, when $r = 0$ this summand is equal to $\det(B)$ and, when $r = n$, it is $\det(A)$.
    	
    	Next, by definition, the characteristic polynomial of $\jac \mbf f(\mbf P)$ is given by
    	\begin{align*}
    		\det(\jac \mbf f(\mbf P)-\lambda I) = \sum_{m = 0}^n (-1)^m \lambda^m \sum_{1\leq i_1 < \ldots < i_{n - m}\leq n} M_{i_1i_2\ldots i_{n - m}}
    	\end{align*}
    	The eigenvalues of $\jac \mbf f(\mbf P)$ can be obtained by equating this polynomial to zero, which can be multiplied by $(-1)^{n}$ to get
    	\begin{align*}
    		\sum_{m = 0}^n (-1)^{n - m} \, \lambda^m \sum_{1 \leq i_1 < \ldots < i_{n - m}\leq n} M_{i_1i_2\ldots i_{n-m}}=0.
    	\end{align*}
    	Now, by the Routh-Hurwitz criterion \cite{Routh1877,Hurwitz1895}, a necessary condition for all the roots of this equation to lie in the left-half plane \evs{is}
    	\begin{align*}
    		(-1)^{k} \sum_{1\leq i_1<i_2<\ldots <i_k \leq n} M_{i_1i_2\ldots i_k} >0
    	\end{align*}
    	for each integer $1\leq k\leq n$, which concludes the proof.
    \end{proof}
    
    The following result \evs{provides} a useful characterisation of the characteristic polynomial in the presence of diffusion. 
    \begin{lemma} \label{detlemma}
    	Let $\mu\in \mathbb R$. Then
    	\begin{align}
    		\det\left(\jac \mbf f(\mbf P) - \mu \, \mathbb D\right) &= b_n \, \mu^n + b_{n - 1} \, \mu^{n - 1} + b_{n - 2} \, \mu^{n - 2} + \ldots + b_1 \, \mu + b_0, \label{eq:mupolynomial}
    	\end{align}
    	where
    	\begin{align*}
    		b_{m} &= (-1)^m \left(\sum_{1 \leq i_1 < \ldots < i_{n - m}\leq n} M_{i_1\ldots i_{n - m}} \prod_{\substack{j = 1 \\ j\not \in \{i_1, \ldots, i_{n - m}\}  }}^n D_j \right), \quad \text{for } 1\leq m< n,
    	\end{align*}
    	$b_0= \det(\jac \mbf f(\mbf P))$, and $b_n=(-1)^n \prod_{j=1}^n D_j$.
    \end{lemma}
    
    \begin{proof}
    	Let us take $A = \jac \mbf f(\mbf P)$, and $B = -\mu \, \mathbb D$ in equation \eqref{determinantofsum}. In particular, in the notation of that equation, note that if $\bs \alpha\neq \bs \beta$, then a row or column of $B(\bs \alpha\vert \bs \beta)$ would be equal to zero as $B$ is a diagonal matrix. Thus the only non-zero contributions to the sum are gotten when $\bs \alpha = \bs \beta$. Thus, notice that for $r = n - m$, and $\bs \beta = \left(i_1, i_2, \ldots, i_{n - m}\right)$, we have that
    	\begin{multline*}
    		\sum_{\bs \alpha,\bs \beta} (-1)^{s(\bs \alpha) + s(\bs \beta)} \det(A[\bs \alpha \vert \bs \beta])\det(B(\bs \alpha \vert \bs \beta))
    		\\ 
    		= (-1)^{2s(\bs \alpha)}\left(\sum_{1 \leq i_1 < \ldots < i_{n - m}\leq n} M_{i_1 \ldots i_{n - m}}\prod_{\substack{j = 1 \\ j\not \in \{i_1, \ldots, i_{n - m}\} }}^n \left[(-1) \, D_j \, \mu\right]\right)
    		\\ 
    		= (-1)^m \left(\sum_{1\leq i_1 < \ldots < i_{n - m}\leq n} M_{i_1 \ldots i_{n - m}}\prod_{\substack{j = 1 \\ j\not \in \{i_1, \ldots, i_{n - m}\} }}^n D_j\right) \mu^m,
    	\end{multline*}
    	which is exactly what is required to conclude the proof.
    \end{proof}
    
    The above \evs{lemma}, enables us to establish the following result, which is a generalisation of Theorem 1 in \cite{Sat}. Indeed, we do not require the Jacobian matrix $\jac \mbf f(\mbf P)$ to be s-stable for this result to be true. Nevertheless, the following Theorem states something only about Turing bifurcations, not considering wave instabilities.
    \begin{theorem} \label{th:noturing}
    	Let $\mbf P$ be a stable homogeneous steady state of \eqref{generalsystem} in the absence of diffusion. If $(-1)^k M_{i_1 i_2\ldots i_k}\geq 0$ for all $k = 1, 2, \ldots, n - 1$, and each increasing word $\left(i_1, i_2, \ldots, i_k\right)$, then $\mbf P$ will not be able to undergo Turing bifurcation when we take into account the full system \eqref{generalsystem}.
    \end{theorem}
    \begin{proof}
    	Note that if $\mbf P$ is a stable homogeneous steady state of \eqref{generalsystem} in the absence of diffusion, then we \evs{have $(-1)^n \det(\jac \mbf f(\mbf P)) > 0$}.
    	
    	Then, observe that the characteristic polynomial of the Jacobian matrix of our system after the addition of diffusion is given by
    	\begin{align*}
    		p_{\mathbb D}(\lambda, \mu) = \det(\jac \mbf f(\mbf P) - \mu \, \mathbb D - \lambda I),
    	\end{align*}
    	where $\mu = |k|^2 > 0$.
    	
    	Since we want to determine conditions in which we can see a Turing instability pattern around $\mbf P$, we want to look for conditions in which there exists a real number $\mu > 0$ such that the characteristic equation has a root $\lambda = 0$. That is, we want to find $\mu > 0$ such that:
    	\begin{align}
    		p_{\mathbb D}(0, \mu) = \det(\jac \mbf f(\mbf P) - \mu \, \mathbb D) = 0, \label{poleqnD}
    	\end{align}
    	which is a polynomial equation in $\mu$, whose coefficients are described in Lemma \ref{detlemma}. Next, notice that if $(-1)^k M_{i_1i_2\ldots i_k}\geq 0$ for all $1\leq k\leq n - 1$ and every increasing word $\left(i_1, i_2, \ldots, i_k\right)$, then all the coefficients in the polynomial \eqref{poleqnD} have the same sign. Therefore, there cannot exist a positive real root of \eqref{poleqnD} for any choice of diffusion constants. Thus, there is no choice of diffusion matrix $\mathbb D$ such that $\mbf P$ admits a Turing bifurcation.
    \end{proof}
    The next two \evs{lemmas} will be crucial to the method we use to find conditions under which either Turing or wave instability patterns can occur in \eqref{generalsystem}.
    \begin{lemma} \label{positiveeig}
    	Let $\mbf P$ be a homogeneous steady state of \eqref{generalsystem}. Assume that there exists a natural number $1 \leq k\leq n$ such that $D_\ell = 0$ for all $\ell \in \{i_1, i_2, \ldots, i_k\}$, while all the other diffusion rates are positive. Then, there exist $k$ eigenvalues of $\jac \mbf f(\mbf P) - \mu \, \mathbb D$ that tend to the eigenvalues of $S_{i_1, i_2, \ldots, i_k}$ as $\mu \to \infty$, while the real part of the others tend to $-\infty$ as $\mu \to \infty$.
    \end{lemma}
    \begin{proof}
    	First, notice that if $k = n$, then all diffusion rates would be equal to 0 and the eigenvalues of $S_{i_1, i_2, \ldots, i_k}$ would be those of $\jac \mbf f(\mbf P)$ which are independent of $\mu$, thus making the result trivially true in this case.
    	
    	Next, if $k < n$, although we might not be able to find all the eigenvectors of $\jac \mbf f(\mbf P) - \mu \, \mathbb D$, note that we can get all its eigenvalues by solving the following equation:
    	\begin{align*}
    		\jac \mbf f(\mbf P) \, \mbf v_m - \mu \, \mathbb D \, \mbf v_m - \lambda_m \, \mbf v_m=0,
    	\end{align*}
    	where $\lambda_m$ is an eigenvalue of $\jac \mbf f(\mbf P) - \mu \, \mathbb D$, and $\mbf v_m$ is a corresponding eigenvector, for each $m = 1, 2, \ldots, n$. Let us denote the elements of $\jac \mbf f(\mbf P)$ and $\mbf v_m$ as follows:
    	\begin{align*}
    		\jac \mbf f(\mbf P) = \left(a_{ij}\right)_{1\leq i, j\leq n}, \qquad 
    		\mbf v_m = \left(v_{m, i}\right)_{1 \leq i \leq n}.
    	\end{align*}
    	Therefore, the eigenvalue problem is equivalent to the following system of equations
    	\begin{align}
    		\sum_{j=1}^n \left(a_{ij} \, v_{m, j}\right) - \mu \, D_i \, v_{m, i} &= \lambda_m \, v_{m,i}, \label{eigensystem}
    	\end{align}
    	for each $i = 1, 2, \ldots, n$. In particular, notice that for each $i = \ell\in \{i_1, i_2, \ldots, i_k\}$, \eqref{eigensystem} becomes
    	\begin{align}
    		\sum_{j = \ell}^n a_{\ell j} \, v_{m, j} = \lambda_m \, v_{m, \ell}. \label{magic}
    	\end{align}
    	Then, at least one of the following two possibilities occurs:
    	\begin{enumerate}
    		\item There exists an eigenvector $\mbf v_p$, with $p \in \{1, 2, \ldots, n\}$, and $q\in \{i_1, i_2, \ldots, i_k\}$ such that $\ds{\lim_{\mu \to \infty} v_{p, q}\neq 0}$. In this case, by continuity, $v_{p, q} \neq 0$ for $\mu$ sufficiently large, in which case \eqref{magic} implies we can write
    		\begin{align}
    			\lambda_p = \frac{1}{v_{p, q}} \sum_{j = 1}^n a_{qj} \, v_{p, j}. \label{usefulcond}
    		\end{align}
    		Also, for $i\not \in \{i_1, i_2, \ldots, i_k\}$, \eqref{eigensystem} implies that
    		\begin{align}
    			\sum_{j = 1}^n \left(a_{ij} \, v_{p, j}\right) - \mu \, D_i \, v_{p, i} &= \frac{v_{p, i}}{v_{p, q}} \sum_{j=1}^n \left(a_{qj} \, v_{p, j}\right). \label{eiglarge}
    		\end{align}
    		Therefore, without loss of generality, if we assume that $\lVert \mbf v_p\rVert = 1$, then everything on the left-hand side of \eqref{eiglarge}, except possibly $\mu$, is finite. Therefore, if we take the limit $\mu\to \infty$, we must have that $\ds{\lim_{\mu\to \infty} v_{p, i} = 0}$ for all $i\not \in \{i_1, i_2, \ldots, i_k\}$. With this condition, for $\mu$ sufficiently large, we know that \eqref{magic} can be simplified to
    		\begin{align*}
    			\sum_{\substack{j\in \{i_1, i_2, \ldots, i_k\}}}^n a_{\ell j} \, v_{p,j} &= \lambda_p \, v_{p, \ell},
    		\end{align*}
    		which can be rewritten as 
    		\begin{align*}
    			S_{i_1i_2\ldots i_k} \, \mbf{\hat v}_p &= \lambda_p \, \mbf{\hat v}_p, \qquad \mbox{where} \quad \mbf{ \hat{v}}_p = \left(v_{p, j}\right)_{j\in \left\{i_1, i_2, \ldots, i_k\right\}}.
    		\end{align*}
    		This means that $\lambda_p$ tends to an eigenvalue of $S_{i_1, i_2, \ldots, i_k}$, as $\mu\to \infty$, and $\mbf{\hat v}_p$ to an associated corresponding eigenvector of that matrix. Note \evs{that} this gives us information only about $k$ eigenvalues, although due to the possibility of generalised eigenvectors, we may not be able to find $k$ independent eigenvectors in this way.
    		\item There exists an eigenvector $\mbf v_p$, with $p \in \{1, 2, \ldots, n\}$, such that $\ds{\lim_{\mu \to \infty} v_{p, q} = 0}$ for all $q\in \{i_1, i_2, \ldots, i_k\}$. \evs{Thus, for such $p$, $i \notin \{i_1, i_2, \ldots, i_k \}$ and $\mu$ sufficiently large, we obtain the set of equations}
    		\begin{align*}
    			\sum_{\substack{j = 1 \\ j\not \in \{i_1, i_2, \ldots, i_k\}}}^n \left(a_{ij} \, v_{p, j}\right) - \mu \, D_i \, v_{p, i} = \lambda_p \, v_{p,i}.
    		\end{align*}
    		This means that $\lambda_p$ is an eigenvalue of the submatrix of $\jac \mbf f(\mbf P) - \mu \, \mathbb D$ formed by the rows and columns that are complementary to $i_1, i_2, \ldots, i_k$. Therefore, because the diagonal elements of this matrix tend to $-\infty$ as $\mu\to \infty$ and the other elements are fixed, then by the Gershgorin Circle Theorem, the eigenvalue $\lambda_p$ must tend to $-\infty$ as $\mu\to \infty$. Again, the number of independent eigenvectors $\mbf v_p$ that fulfill the condition of this case depend on the number of generalized eigenvectors of the submatrix in question. 
    	\end{enumerate}
    	\evs{Finally, the first} (respectively, second) possibility gives us information about only $k$ (respectively, $n - k$) eigenvalues. Therefore, both possibilities must occur simultaneously to describe the complete set of eigenvalues of $\jac \mbf f(\mbf P) - \mu \, \mathbb D$.
    \end{proof}
    
    \begin{lemma} \label{nonpositiveeig}
    	Let $\mbf P$ be a homogeneous steady state of \eqref{generalsystem}. If we assume that there exists a natural number $1\leq k \leq n$ such that $D_\ell \to \infty$ for all $\ell\in \left\{i_1, i_2, \ldots, i_k\right\}$, while all the other diffusion rates are positive but finite, then there exist $k$ eigenvalues of $\jac \mbf f(\mbf P) - \mu \, \mathbb D$ having real parts that tend to $-\infty$ for every $\mu > 0$, while $n - k$ eigenvalues of that matrix tend to the eigenvalues of the submatrix of $\jac \mbf f(\mbf P)$ formed out of rows and columns that are complementary to $i_1, i_2, \ldots, i_k$, as $\mu \to 0^+$.
    \end{lemma}
    
    \begin{proof}
    	Using the same notation as in the previous proof, we need to solve the equations \eqref{eigensystem} for each $i = 1, 2, \ldots, n$. Again, in this case, we have two possibilities (note that the proofs are analogous but have some key differences in the sets of indices):
    	\begin{enumerate}
    		\item There exist an eigenvector $\mbf v_p$, with $p\in \{1, 2, \ldots, n\}$ and \evs{$q \notin \{i_1, i_2, \ldots, i_k\}$} such that $\ds{\lim_{D_{i_1}, \ldots, D_{i_k}\to \infty} v_{p, q}\neq 0}$ for a fixed $0 <\mu\ll 1$ sufficiently small.	In this case, \eqref{eigensystem} implies that we can write 
    		\begin{align}
    			\lambda_p = \frac{1}{v_{p, q}}\left(\sum_{j = 1}^n a_{qj} \, v_{p, j} - \mu \, D_q \, v_{p, q}\right) \label{usefulcond2}
    		\end{align}
    		and, for $i\in \{i_1, i_2, \ldots, i_k\}$,
    		\begin{align}
    			\sum_{j = 1}^n \left(a_{ij} \, v_{p, j}\right) - \mu \, D_i \, v_{p, i} &= \frac{v_{p, i}}{v_{p, q}} \left(\sum_{j = 1}^n \left(a_{lj} \, v_{p, j}\right) - \mu \, D_q \, v_{p, q}\right). \label{eiglarge2}
    		\end{align}
    		If we assume, without loss of generality, that $\lVert \mbf v_p\rVert = 1$, then everything on the right-hand side of \eqref{eiglarge2} is finite whilst everything on the left-hand side is finite, except possibly $\mu \, D_i$. Now if we let $D_i \to \infty$, while $\mu$, sufficiently small, remains fixed. Then, we must have that $\ds{\lim_{D_{i_1}, \ldots, D_{i_k}\to \infty} v_{p, i} = 0}$, for all $i \in \{i_1, i_2, \ldots, i_k\}$, and $0 < \mu \ll 1$ sufficiently small. With this condition, for $0 < \mu \ll 1$, we know that for $m = p$, \eqref{eigensystem} can be simplified to
    		\begin{align*}
    			\sum_{\substack{j = 1 \\ j \not \in \{i_1, i_2, \ldots, i_k\}}}^n a_{ij} \, v_{p, j} = \lambda_p \, v_{p, i},
    		\end{align*}
    		for each $i\not \in \{i_1, i_2, \ldots, i_k\}$. As before, this means that $\lambda_p$ tends to an eigenvalue of the submatrix of $\jac \mbf f(\mbf P)$ that takes into account the rows and columns that are complementary to $i_1, i_2, \ldots, i_k$, for $0<\mu\ll 1$ sufficiently small. \evs{Note that} this gives us information only about $n - k$ eigenvalues. Again, we highlight that this does not mean that we will always be able to find $n - k$ eigenvectors that fulfill the condition required for this case. That depends on the number of generalized eigenvectors of the submatrix of $\jac \mbf f(\mbf P)$ that takes into account only its rows and columns that are complementary to $i_1, i_2, \ldots, i_k$.
    		
    		\item There exists an eigenvector $\mbf v_p$, with $p \in \left\{1, 2, \ldots, n\right\}$ such that $\ds{\lim_{D_{i_1}, \ldots, D_{i_k}\to \infty} v_{p, q} = 0}$ for all \evs{$q \notin \{i_1, i_2, \ldots, i_k\}$}, and $0 < \mu \ll 1$ sufficiently small. Therefore, the system of equations we need to solve for $i \in \{i_1, i_2, \ldots, i_k\}$ is given by:
    		\begin{align*}
    			\sum_{\substack{j = 1 \\ j \in \left\{i_1, i_2, \ldots, i_k\right\}}}^n \left(a_{ij} \, v_{p, j} \right) - \mu \, D_i \, v_{p, i} &= \lambda_p \, v_{p,i}.
    		\end{align*}
    		This means that $\lambda_p$ tends to an eigenvalue of the submatrix of $\jac \mbf f(\mbf P) - \mathbb D \, \mu$ formed out of rows and columns $i_1, i_2, \ldots, i_k$. Therefore, as the off-diagonal elements of that matrix are constant, while the diagonal entries tend to $-\infty$, then by the Gerschgorin circle theorem, the real part of $\lambda_p$ tends to $-\infty$ as $D_\ell\to \infty$, for each $\ell \in \left\{i_1, i_2, \ldots, i_k\right\}$, and $0 < \mu \ll 1$. Furthermore, note that as some diffusion rates are infinitely large, then increasing the value of $\mu$ above 1 will not impede the use of the Gerschgorin circle theorem to say that, in this case, the real part of the eigenvalues will tend to $-\infty$ for every $\mu > 0$.
    	\end{enumerate}
    	Finally, as the first (resp. second) possibility gives us information only about $n - k$ (resp. $k$) eigenvalues, then they both must occur in order to describe the $n$ eigenvalues of $\jac \mbf f(\mbf P) - \mu \, \mathbb D$.
    \end{proof}
    Now we introduce a useful definition and an important general result that enables us to find specific ratios of diffusion constants that give rise to an instability.
    \begin{definition}        
    	For each $\ell \in \mathbb Z^+$, let $s$ be the \textit{sorting function}
    	\begin{align*}
    		s : V\subset (\mathbb Z^+)^\ell &\to (\mathbb Z^+)^\ell
    		\\
    		\left(j_1, j_2, \ldots, j_\ell\right) &\to s\left(j_1, j_2, \ldots, j_\ell\right) = \left(j_{\sigma(1)}, j_{\sigma(2)}, \ldots, j_{\sigma(\ell)}\right),
    	\end{align*}
    	where $V\subset [1\ldots n]^\ell$ is a subset of vectors with no repeated entries, and $\left(j_{\sigma(1)}, j_{\sigma(2)}, \ldots, j_{\sigma(\ell)}\right)$ is an increasing word.
    \end{definition}
    
    \begin{lemma} \label{signpreservation}
    	Let $\left(j_1, j_2, \ldots, j_n\right)$ be a word comprising non-repeated integers between $1$ and $n$ (not necessarily increasing), such that $M_{s\left(j_1,j_2,\ldots ,j_k\right)}\neq 0$ for each integer $1\leq k\leq n$. Then we can choose positive diffusion constants $D_i$, $i=1,\ldots, n$ such that
    	\begin{align*}
    		\mbox{sign}\left(b_k\right) = \mbox{sign}\left((-1)^k M_{s\left(j_1, j_2, \ldots, j_k\right)}\right),
    	\end{align*}
    	for each integer $1\leq k\leq n$. Furthermore, this result is invariant under multiplication of all the diffusion constants by a positive real number.
    \end{lemma}
    
    \begin{proof}
    	First, let $\varepsilon > 0$ and, for each integer $1 \leq k \leq n$, take $D_{j_k} = \varepsilon^{n - k} > 0$. \evs{Then, for each integer} $1\leq \ell \leq n$, the terms $(-1)^\ell M_{s (i_1i_2\ldots i_\ell )}$ in \eqref{eq:mupolynomial} will have coefficients of the form $\varepsilon^{p\left(\hat i_1, \hat i_2, \ldots, \hat i_{n - \ell}\right)}$, where $p\left(\hat i_1, \hat i_2, \ldots, \hat i_{n-\ell}\right)$ is the sum of the powers determined by the diffusion rates $D_{\hat i_q}$, where $\hat i_q \not \in \left\{i_1, i_2, \ldots, i_\ell\right\}$ for each integer $1\leq q \leq n - \ell$. Next, notice that for $1\leq m\leq n$, we have that
    	\begin{multline*}
    		b_m = (-1)^m \left(\sum_{1\leq i_1 < \ldots < i_{n - m}\leq n} M_{i_1 \ldots i_{n - m}} \prod_{\substack{j = 1 \\ j\not \in \{i_1, \ldots, i_{n - m}\}}}^n D_j \right)
    		\\
    		= (-1)^m \left(\sum_{\substack{1\leq i_1 < \ldots < i_{n - m}\leq n \\ \left(i_1, \ldots, i_{n - m}\right)\neq s\left(j_1, j_2, \ldots, j_{n - m}\right)}} M_{i_1\ldots i_{n - m}} \varepsilon^{p\left(\hat i_1, \ldots, \hat i_m\right)}\right)
    		\\ 
    		+ (-1)^m \varepsilon^{p\left(j_{n - m + 1}, j_{n - m + 2}, \ldots, j_n\right)} M_{s\left(j_1, j_2, \ldots, j_{n - m}\right)},
    	\end{multline*}
    	where $\hat i_\ell \not \in \left(i_1, i_2, \ldots, i_{n - m}\right)$ for every integer $1 \leq \ell \leq m$.
    	
    	Here, note that the lowest power of $\varepsilon$ will be $p\left(j_{n - m + 1}, j_{n - m + 2}, \ldots, j_n\right)$. That coefficient will be the sum of all the lowest exponents of $\varepsilon$. All the other terms are multiplied by the same number of diffusion rates but have at least a difference in one index with the diffusion rates multiplying $(-1)^m M_{s\left(j_1, j_2, \ldots, j_{n - m}\right)}$. This means that the other terms will have a higher order in terms of $\varepsilon$. With this, taking $\varepsilon > 0$ sufficiently small, we will have that
    	\begin{align*}
    		\mbox{sign}\left(b_k\right) = \mbox{sign}\left((-1)^k M_{s\left(j_1, j_2, \ldots, j_k\right)}\right),
    	\end{align*}
    	for every integer $1\leq k\leq n$.
    	
    	Last but not least, note that if we pick any $\delta > 0$ and perform the change $D_i \to \delta D_i$ for each integer $1\leq i\leq n$, then 
    	\begin{align*}
    		b_m &= (-1)^m \, \delta^m \left(\sum_{1 \leq i_1 < \ldots < i_{n - m}\leq n} M_{i_1 \ldots i_{n - m}} \prod_{\substack{j = 1 \\ j \not \in \{i_1, \ldots, i_{n - m}\}}}^n D_j \right),
    	\end{align*}
    	which will not change the sign of these terms. 
    \end{proof}
    
    We close this section by stating an important result which is a trivial extension from \cite[Theorem 1]{Sat}
    \begin{definition}
        The matrix $\jac \mbf f(\mbf P)$ is said to be s-stable if all its principal submatrices have all their eigenvalues in the left half-plane. 
    \end{definition}
    \begin{theorem}[\cite{Sat}]
    \label{th:sat}
        If $\jac \mbf f(\mbf P)$ is s-stable, then no Turing or wave instabilities are possible from the homogeneous steady state $\mbf P$ for any choice of diffusion constants $D_i \geq 0$. 
    \end{theorem}
    Note that the statement of this theorem in \cite{Sat} only mentions Turing instability, not wave instability. Nevertheless, it is clear from the proof that s-stability implies that there can be no eigenvalue of $\jac \mbf f(\mbf{P}) -\mu \, \mathbb{D}$ with a positive real part.

\section{Wave instability for $n\leq 3$} \label{sec:Hopf}
    In this section, we shall try to understand the minimal ingredients for wave instabilities to occur. Essentially, this result was originally stated in \cite[Th.~1.1]{Anma}, but we noted that the theorem contains some slightly stricter conditions on the diffusion constants $D_i$ and also includes a claim (result (iv) in the theorem) which is false in general. For example, in the notation of the present paper, the matrix
    $$
    	\jac \mbf f(\mbf P) = \begin{pmatrix}
            -4 & -4 & -4
            \\
            \phantom{-}3 & \phantom{-}2 & \phantom{-}2
            \\
            \phantom{-}1 & \phantom{-}1 & -1
        \end{pmatrix}
    $$
    can easily be shown to provide a counterexample when taking $D_1 = 3$ and $D_2, D_3$ small. Finally, the general way we state our result is easily generalisable to arbitrary $n \geq 3$, which forms the subject of the following Section.
    
    First, we state our general result.
    \begin{proposition}\label{th:hopf3}
    	Let $\mbf P$ be a stable homogeneous steady state of \eqref{generalsystem} in the absence of diffusion. Then the minimum number of components that the system needs to develop wave instabilities around $\mbf P$ is 3. Furthermore, if $n = 3$ and there exist $1 \leq i_1 < i_2 \leq 3$ such that $M_{i_1} + M_{i_2} > 0$ and $M_{i_1i_2} > 0$, then we can choose diffusion rates $D_1, D_2, D_3 \geq 0$ such that \eqref{generalsystem} develops wave instabilities at $\mbf P$.
    \end{proposition}
    Before proceeding \evs{with} a proof, we establish a few useful standard lemmas.
    \begin{lemma}\label{n2Hopf}
    	Let $n = 2$. Suppose that $\mbf P$ is a stable equilibrium point of \eqref{generalsystem} in the absence of diffusion. Then the system will not admit a wave instability around $\mbf P$ for any diffusion matrix ${\mathbb D}$.
    \end{lemma}
    \begin{proof}
    	Let us label the elements of $\jac \mbf f(\mbf P)$ as follows
    	\begin{align*}
    		\jac \mbf f(\mbf P) = \begin{pmatrix}
    			\alpha & \beta
    			\\
    			\gamma & \delta
    		\end{pmatrix}.
    	\end{align*}
    	Given that all the eigenvalues of $\jac \mbf f(\mbf P)$ are in the left-half plane, then its trace must be negative. That is $\alpha + \delta < 0$. To have a wave instability around $\mbf P$, we need to find $\mu > 0$ such that the eigenvalues of $\jac \mbf f(\mbf P) - \mu \, \mathbb D$ are purely imaginary. We have that
    	\begin{align*}
    		\tr(\jac \mbf f(\mbf P) - \mu \, \mathbb D) &= \tr(\jac \mbf f(\mbf P)) - \left(D_1 + D_2\right) \mu < 0,
    		\\ 
    		\det(\jac \mbf f(\mbf P) - \mu \, \mathbb D) &= \det(\jac \mbf f(\mbf P)) - \left(D_2 \, \alpha + D_1 \, \delta\right) \mu + D_1\, D_2 \, \mu^2
    	\end{align*}
    	This means that the eigenvalues of $\jac \mbf f(\mbf P)-\mu \, \mathbb D$ are given by
    	\begin{align}
    		\lambda_{1, 2} = \frac{\tr(\jac \mbf f(\mbf P)) - \left(D_1 + D_2\right) \mu \pm \sqrt{\tr(\jac \mbf f(\mbf P) - \mu \, \mathbb D)^2-4\det(\jac \mbf f(\mbf P)-\mu \, \mathbb D)}}{2}, \label{eq:l12}
    	\end{align}
    	Finally, note that in case that $\lambda_{1, 2}$ are complex conjugate, then \eqref{eq:l12} shows that $\Re(\lambda_{1, 2}) < 0$ for all $\mu > 0$, which concludes the proof. 
    \end{proof}
    Next, let us consider the case $n = 3$, and use the notation for principal minors established in Definition \ref{def:minors}. The following two \evs{lemmas} establish a general condition for finding non-negative diffusion rates for a wave instability, and a sufficient condition under which it cannot occur. 
    \begin{lemma} \label{th:bdef3}
    	Let $n = 3$. Then the condition for the matrix $\jac \mbf f(\mbf P) - \mu \, \mathbb D$ to have a purely imaginary pair of eigenvalues for some $\mu > 0$ can be expressed as
    	\begin{equation}
    		b_3 \, \mu^3 + b_2 \, \mu^2 + b_1 \, \mu + b_0 = 0, \label{complexpoleq3D}
    	\end{equation}
    	where
    	\begin{align*}
    		b_3 &= \left(D_1 + D_2\right) \left(D_1 + D_3\right) \left(D_2 + D_3\right),
    		\\ 
    		b_2 &= - D_1^2 \left(M_2 + M_3\right) - D_2^2 \left(M_1 + M_3\right) - D_3^2 \left(M_1 + M_2\right) 
    		\\
    		& \quad - 2\left(D_1 \, D_2 + D_1 \, D_3 + D_2 \, D_3\right)\left(M_1 + M_2 + M_3\right),
    		\\ 
    		b_1 &= D_1 \left(\left(M_2 + M_3 \right) \left(M_1 + M_2 + M_3\right) + M_{12} + M_{13}\right)
    		\\
    		& \quad + D_2 \left(\left(M_1 + M_3 \right) \left(M_1 + M_2 + M_3 \right) + M_{12} + M_{23}\right)
    		\\ 
    		&\quad + D_3 \left(\left(M_1 + M_2\right) \left(M_1 + M_2 + M_3 \right) + M_{13} + M_{23}\right),
    		\\ 
    		b_0 &= - \left(M_1 + M_2 + M_3\right)\left(M_{12} + M_{13} + M_{23}\right) + \det(\jac \mbf f(\mbf P)).
    	\end{align*}
    \end{lemma}
    \begin{proof}
    	The characteristic polynomial of $\jac \mbf f(\mbf P) - \mu \, \mathbb D$ can be written as
    	\begin{align}
    		p_D(\lambda, \mu) := \lambda^3 + a_2 \, \lambda^2 + a_1 \, \lambda + a_0, \label{eq:PD3}
    	\end{align}
    	where
    	\begin{align*}
    		a_2 &= -\left(M_1 + M_2 + M_3\right) + \left(D_1 + D_2 + D_3\right) \mu,
    		\\ 
    		a_1 &= M_{13} + M_{12} + M_{23} - \left(D_1\left(M_2 + M_3\right) + D_2\left(M_1 + M_3 \right) + D_3\left(M_1 + M_2\right)\right) \mu 
    		\\ 
    		&\quad + \left(D_1 \, D_2 + D_1 \, D_3 + D_2 \, D_3\right) \mu^2
    		\\ 
    		a_0 &= - \det(\jac \mbf f(\mbf P) - \mu \, \mathbb D).
    	\end{align*}
    	Therefore, when we evaluate this characteristic polynomial at $\lambda = \omega \, i$, we will have the following expression
    	\begin{align*}
    		p_D(\omega \, i, \mu) &:= - \omega^3 \, i - a_2 \, \omega^2 + a_1 \, \omega \, i + a_0 = \omega \left(- \omega^2 + a_1\right) i + \left(- a_2 \, \omega^2 + a_0\right).
    	\end{align*}
    	Setting this expression to zero we get two expressions for $\omega^2$ that are well defined if and only if $a_1>0$, $a_2\neq 0$, and $a_1a_2-a_0=0$. Upon substitution of the definitions of $a_0,a_1,a_2$ into this final equality, we arrive at the polynomial in $\mu$ given \evs{at} the statement of the result.
    \end{proof}

    \begin{lemma}\label{hopf3sufficient}
    	Let $n = 3$. If $M_{i_1} + M_{i_2}\leq 0$ for all $1 \leq i_1 < i_2 \leq 3$, $M_{12} + M_{13}\geq 0$, $M_{12} + M_{23}\geq 0$, $M_{13} + M_{23}\geq 0$, and $\det(\jac \mbf f(\mbf P)) \leq 0$, then the system will not be able to show a wave instability at $\mbf P$.
    \end{lemma}
    
    \begin{proof}
    	The proof is clear because, under the hypotheses, the coefficients in Lemma \ref{th:bdef3}, $b_0, b_1, b_2, b_3$, will be non-negative. This means that \eqref{complexpoleq3D} cannot have a positive real solution for $\mu$.
    \end{proof}
    
    \begin{remark}
      The previous \evs{lemma} says that if \evs{all $2 \times 2$ principal submatrices of the linearisation around a homogeneous steady state of a $3$-component system have eigenvalues with negative real part}, then the system will not be able to undergo a wave instability. \evs{Proposition \ref{th:hopf3} states} the converse, namely that if there is an unstable $2\times2$ principal submatrix \evs{having} two eigenvalues with a positive real part, then there will exist diffusion rates for which there is a wave instability.
    \end{remark}
    
    \begin{lemma}\label{hopf3necessary}
    	Let $n = 3$ and suppose $\mbf P$ be a stable equilibrium point of \eqref{generalsystem} in the absence of diffusion, then $M_1 + M_2 + M_3 < 0$, $M_{12} + M_{13} + M_{23} > 0$, $\det(\jac \mbf f(\mbf P)) < 0$, and $b_0 > 0$.
    \end{lemma}
    
    \begin{proof}
    	Most of the statement is a trivial consequence of \evs{Lemma} \ref{necessarystability} for the case $n = 3$. The only thing that remains to be proved is that $b_0 > 0$. Upon writing the characteristic polynomial of $\jac \mbf f(\mbf P)$ as $p(\lambda) = \lambda^3 + c_2 \, \lambda^2 + c_1 \, \lambda + c_0$. Then the Routh-Hurwitz criterion ensures that $\mbf P$ is stable if and only if $c_0, c_1, c_2 > 0$ and $(c_1 \, c_2 - c_0)/c_2 > 0$. This final condition can be rearranged to read 
    	$$
    	   b_0 = -\left(M_1 + M_2 + M_3\right)\left(M_{13} + M_{12} + M_{23}\right) + \det(\jac \mbf f(\mbf P)) > 0.
    	$$
    \end{proof}
    
    Finally, we can prove Proposition \ref{th:hopf3}.
    
    \begin{proof}[Proof of Proposition \ref{th:hopf3}]
    	Lemma \ref{n2Hopf} shows that wave instabilities cannot occur for $n = 2$. Therefore, let $n = 3$. Lemma \ref{hopf3sufficient} shows that no wave instability can occur if $\jac \mbf f(\mathbb P)$ has only $2\times 2$ principal submatrices with a pair of negative real-part eigenvalues. The condition on principle submatrices stated in the theorem is essentially to assume that there is one principal 2x2 submatrix with a pair of eigenvalues with a positive real part. Without loss of generality, let us assume that $M_2 + M_3 > 0$ and $M_{23} > 0$. Then we can choose $D_2 = D_3 = 0$. With this, note that $b_3 = 0$ and the coefficient $a_1$ defined in \eqref{eq:PD3} is a decreasing function on $\mu$, which is positive when $\mu = 0$ and it is equal to zero when
    	\begin{align*}
    		\mu = \mu^* = \frac{M_{13} + M_{12} + M_{23}}{D_1\left(M_2 + M_3\right)} > 0.
    	\end{align*}
    	Therefore, using the notation of Lemma \ref{th:bdef3}, if we define the function
    	\begin{align*}
    		q(\mu) = b_2 \, \mu^2 + b_1 \, \mu + b_0,
    	\end{align*}
    	then we have that $q(0) = b_0$ \evs{which, by Lemma \ref{hopf3necessary}, is strictly positive, while
    	\begin{align*}
    		q\left(\mu^*\right) = M_{123} - \frac{M_{23} \left(M_{12} + M_{13} + M_{23}\right)}{M_2 + M_3}.
    	\end{align*}
    	is strictly negative}. Therefore, thanks to the Intermediate Value Theorem, there must exist at least one solution to the equation $q(\mu) = 0$ in the interval $0 < \mu < \mu^*$, which concludes the proof.
    \end{proof}

\section{General conditions for Turing and wave bifurcations}\label{sec:general}
    So far, we have established sufficient conditions for the absence of a Turing instability in a system around $\mbf P$ (Theorem \ref{th:noturing}) along with the minimum number of components, a sufficient condition for a wave instability to occur, and Proposition \ref{th:hopf3}. In this section, we shall establish further general conditions for system \eqref{generalsystem} to be able to show a Turing or wave instability around $\mbf P$ for any number of components and show the existence of Turing or wave bifurcations as parameters vary.
    
    First, we are going to state a general standard result that gives us some conditions in which we can have neither a Turing nor a wave instability around $\mbf P$.
    
    \begin{theorem}
    	Let $\mbf P$ be a stable equilibrium point of \eqref{generalsystem} in the absence of diffusion. Suppose that all diffusion rates are equal, that is there exists $\delta \geq 0$ such that $D_i = \delta$ for every integer $1 \leq i \leq n$. Then \eqref{generalsystem} can admit neither a Turing nor a wave instability around $\mbf P$.
    \end{theorem}
    \begin{proof}
    	Under the hypothesis on equal diffusion constants, note that the eigenvalues, $\lambda$, of $\jac \mbf f(\mbf P) - \mu \, \mathbb D$ satisfy the equation
    	\begin{align*}
    		\det(\jac \mbf f(\mbf P) - \left(\mu \, \delta + \lambda\right) I)&=0,
    	\end{align*}
    	which implies that $\mu \, \delta + \lambda$ is an eigenvalue of $\jac \mbf f(\mbf P)$. Therefore, because $\mbf P$ is \evs{stable, $\Re(\mu \, \delta + \lambda) < 0$ implies $\Re (\lambda) < - \mu \, \delta < 0$}.
    	
    	Hence the eigenvalues of $\jac \mbf f(\mbf P)-\mu \, \mathbb D$ cannot cross the imaginary axis for any $\mu > 0$. 
    \end{proof}
    We are now ready to state our main theorem, which develops sufficient conditions for system \eqref{generalsystem} to admit a Turing or wave instability pattern around $\mbf P$; see also Fig.~\ref{fig:dispersionrelationsdependingonD} for the main idea.
    \begin{theorem} \label{biggestresult}
    	Let $\mbf P$ be a stable homogeneous steady state of \eqref{generalsystem} in the absence of diffusion. If there exists an integer $1\leq \ell \leq n - 1$, and a word $\left(j_1, j_2, \ldots, j_n\right)$ of non-repeated integers between $1$ and $n$ (not necessarily increasing), such that $S_{s\left(j_1, j_2, \ldots, j_\ell\right)}$ has $q > 0$ eigenvalues with a positive real part, then if $M_{s\left(j_1, j_2, \ldots, j_k\right)}\neq 0$ for each integer $\ell \leq k \leq n$ and $(-1)^k M_{s\left(j_1, j_2, \ldots, j_k\right)}$ changes its sign $p$ times as $k$ \evs{increases} from $\ell$ to $n$, then:
    	\begin{enumerate}
    		\item $q$ is even if and only if $p$ is even;
    		\item if $p$ is even (respectively, odd), then there exists a choice of non-negative diffusion rates $D_i$, $i = 1,\ldots n$ such that $\jac \mbf f(\mbf P) - \mu \, \mathbb D$ has an even (odd) number of zero eigenvalues, bounded above by $p$ as $\mu$ varies from $0$ to $\infty$;
    		\item furthermore, with those diffusion rates, if $q > p$, then $\jac \mbf f(\mbf P) - \mu \, \mathbb D$  has at least $\dfrac{q - p}{2}$ pairs of complex conjugate eigenvalues with zero real part as $\mu$ varies from $0$ to $\infty$;
    		\item in particular, with the same diffusion rates, the number of crossings to the axis $\Re(\lambda) = 0$ will be at least $q$ as $\mu$ varies from 0 to $\infty$.
    	\end{enumerate}
    \end{theorem}
    \begin{proof}
    	Let $p > 0$ and $q > 0$ be defined as in the statement of the theorem. Note that the definition of $p$ and $q$ is independent of the value of the diffusion constants. Therefore, we can start by setting $D_m = 0$ for all $m = j_1, j_2, \ldots, j_\ell$. Next, thanks to Lemma \ref{signpreservation}, we will be able to find values for the other $n - \ell$ diffusion constants such that the signs of $b_k$ are the same as $(-1)^k M_{s\left(j_1, j_2, \ldots, j_k\right)}$, for each $\ell + 1 \leq k \leq n$.
    	
    	This implies that, at least, $q$ eigenvalues will cross the imaginary axis as $\mu \to \infty$. On the other hand, as the number of complex conjugate eigenvalues that cross this axis is even, \evs{as the polynomial of $\mu$ in equation \eqref{eq:mupolynomial} changes sign $p$ times}, if we assume that $p$ is even (resp. odd), then we will have an even (resp. odd) number of real zeros crossing that axis, which implies that $q$ is even (resp. odd).
    	
    	Next, note that we will have only $p$ changes of sign in the polynomial equation for $\mu$, which implies, by Descartes' rule of signs, that we will have an even (resp. odd) number of values of $\mu>0$ such that $\jac \mbf f(\mbf P) - \mu \, \mathbb D$ has a zero eigenvalue if $p$ is even (resp. odd), while (if $q > p$) at least $\dfrac{q - p}{2}$ values of $\mu>0$ such that $\jac \mbf f(\mbf P) - \mu \, \mathbb D$ has a pair of purely imaginary eigenvalues. Therefore, by a matter of continuity, we will be able to increase the zero-diffusion rates by sufficiently small amounts in order to keep the same number of positive eigenvalues and eigenvalues with a positive real part, even though the number of crossings may increase. This concludes the proof.
    \end{proof}
    \begin{remark}
    	Theorem \ref{biggestresult} effectively characterizes zeros we can find in dispersion curves as $\mu = k^2$ tends to $\infty$. Figure \ref{fig:dispersionrelationsdependingonD} shows one of the simplest cases that can occur in the dispersion relation of a system that has five components. In (a) \ref{fig:zerodiffrates}, three diffusion rates are equal to zero while the others are positive. In that graph, we can recognize immediately that a submatrix $S_{i_1i_2i_3}$ is unstable and has two eigenvalues with a positive real part. Furthermore, when setting $D_{i_1} = D_{i_2} = D_{i_3} = 0$, there is an even number of changes of sign in the polynomial equation in $\mu$. Next, in Fig.~\ref{fig:positivediffrates}, all the diffusion rates are positive, so the eigenvalues that were converging to a fixed value as $k\to \infty$ are now decaying.  Finally, Fig.~\ref{fig:bifurcationdiffrates} shows what happens as we keep increasing the diffusion rates that were equal to zero up to a point in which one of the dispersion relations becomes tangent to the $k$-axis and we find a bifurcation point. It is worth highlighting that these are only depicting the real part of the eigenvalues and the figure could equally well apply to a higher-component system where some of the dispersion curves represent complex conjugate eigenvalues. See also Fig.~\ref{fig:wonderful_example} below.
    \end{remark}
    \begin{figure}
    	\centering
    	\begin{subfigure}[b]{0.32\textwidth}
    		\centering
    		\includegraphics[width=\textwidth]{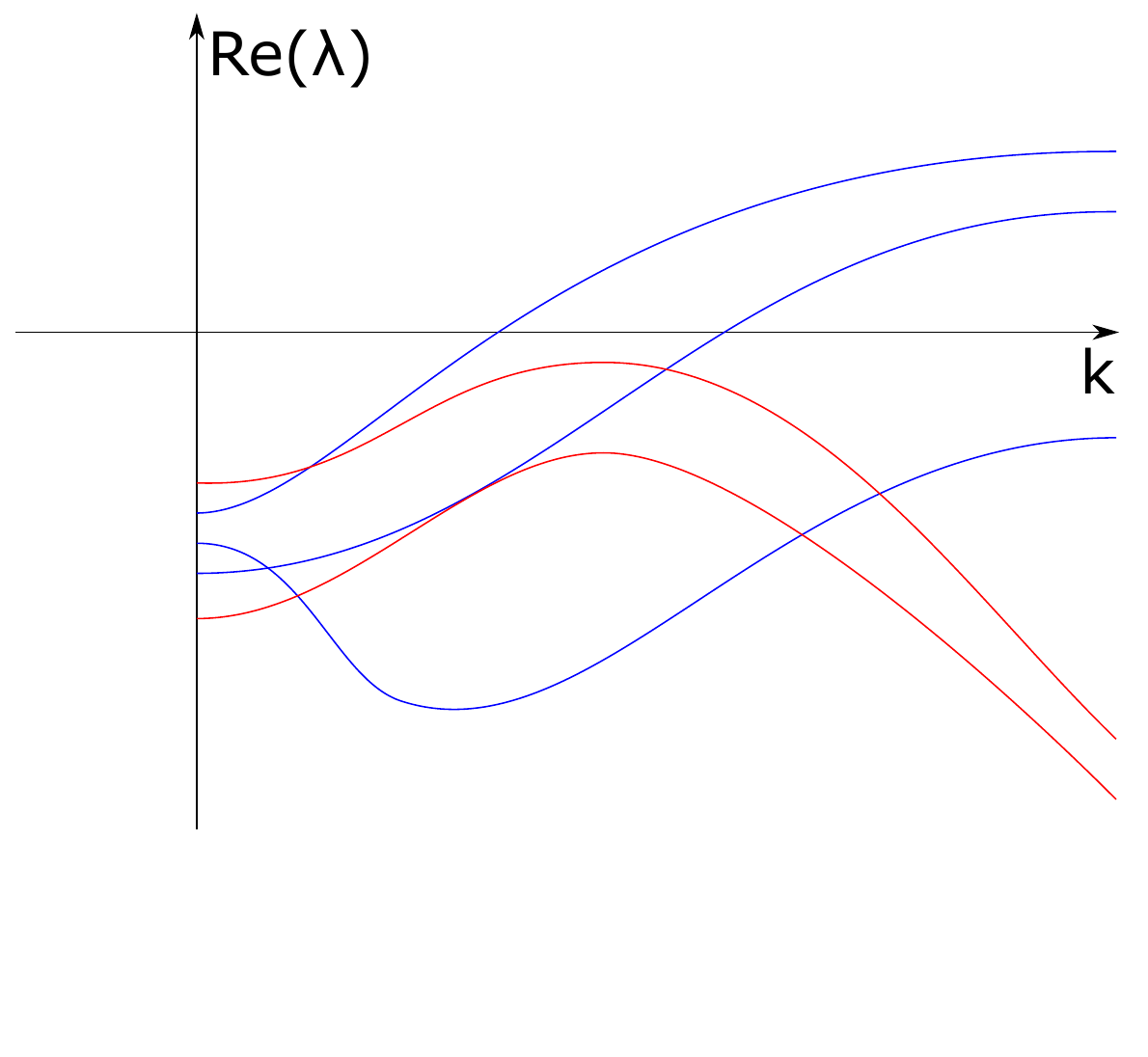}
    		\caption{ }
    		\label{fig:zerodiffrates}
    	\end{subfigure}
    	\hfill
    	\begin{subfigure}[b]{0.32\textwidth}
    		\centering
    		\includegraphics[width=\textwidth]{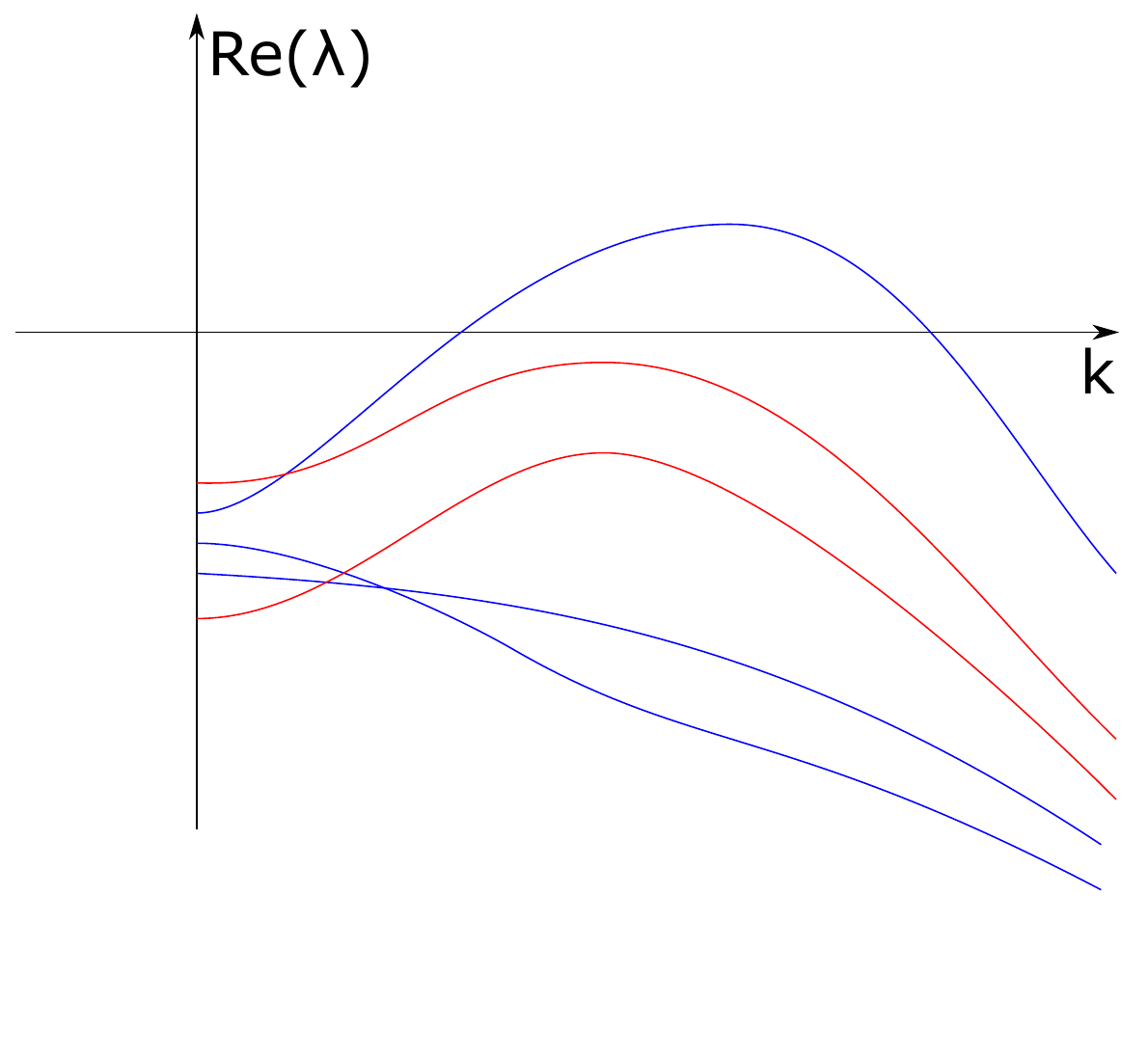}
    		\caption{ }
    		\label{fig:positivediffrates}
    	\end{subfigure}
    	\begin{subfigure}[b]{0.32\textwidth}
    		\centering
    		\includegraphics[width=\textwidth]{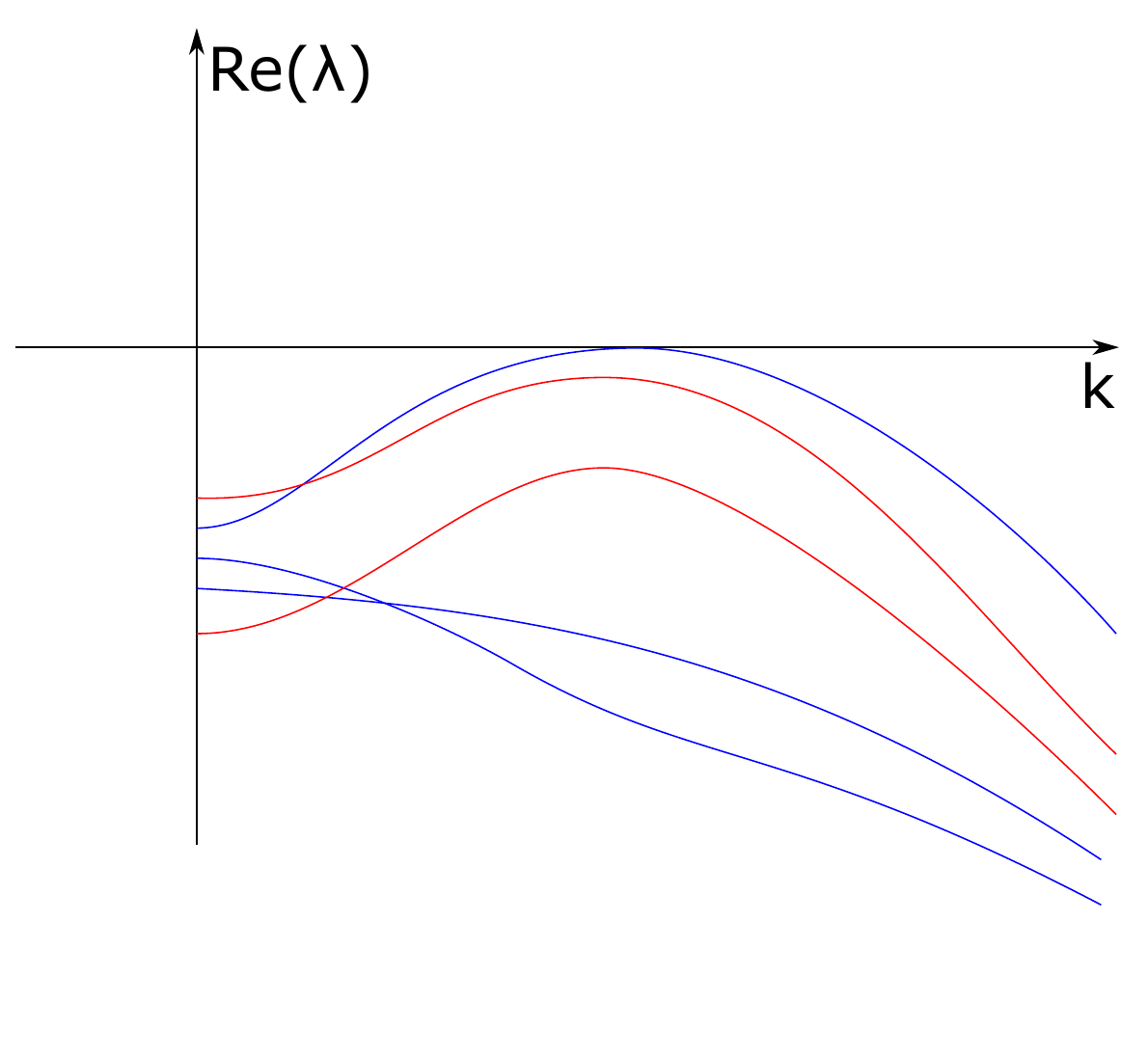}
    		\caption{ }
    		\label{fig:bifurcationdiffrates}
    	\end{subfigure}
    	\caption{Illustrating the main idea of Theorem \ref{biggestresult} via sketches of dispersion curves of a five-component system, for three different sets of values of the diffusion rates. In (a) three diffusion rates (corresponding to the blue curves) are zero, while the others are positive. In (b) all diffusion rates are positive and there is an interval of $k$-values corresponding to a Turing instability. \evs{Panel (c) shows the critical situation between cases (a) and (b), which would correspond to a Turing bifurcation.}}
    	\label{fig:dispersionrelationsdependingonD}
    \end{figure}
    
    A natural consequence of Theorem \ref{biggestresult} is its use in showing instabilities that might not have been apparent if only looking for Turing instabilities. The following is a straightforward consequence of the theorem.
    \begin{corollary}
    	Let $\mbf P$ be a stable homogeneous steady state of \eqref{generalsystem} in the absence of diffusion. If $\mbf P$ does not fulfill the conditions to show Turing instabilities, but one of the principal submatrices $S_{i_1i_2\ldots i_k}$ has an even number of eigenvalues with a positive real part, then we can find non-negative diffusion rates such that $\mbf P$ \evs{exhibits} wave instabilities.
    \end{corollary}
    \begin{remark}
    	To illustrate why the corollary is important, consider the following Jacobian matrix:
    	\begin{align*}
    		\jac \mbf f(\mbf P) = \begin{pmatrix}
    			-32 & -1 & -3 & -2
    			\\ 
    			-3 & -1 & -2 & 1
    			\\ 
    			-1 & 0 & -1 & -2
    			\\ 
    			5 & -2 & 1 & -0.1
    		\end{pmatrix}
    	\end{align*}
    	The eigenvalues of this matrix are given by
    	\begin{align*}
    		\lambda_1 = -31.8681, \quad \lambda_2 = -2.13507, \quad \lambda_{3, 4} = -0.0484039 \pm 2.2506 i.
    	\end{align*}
    	Furthermore, this matrix satisfies
    	\begin{align*}
    		(-1)^1 \, M_1 = 32, \quad (-1)^1 \, M_2 = 1, \quad (-1)^1 \, M_3 = 1, \quad (-1)^1 \, M_4 = 0.1,
    		\\ 
    		(-1)^2 \, M_{12} = 29, \quad (-1)^2 \, M_{13} = 29, \quad (-1)^2 \, M_{14} = 13.2,
    		\\ 
    		(-1)^2 \, M_{23} = 1, \quad (-1)^2 \, M_{24} = 2.1, \quad (-1)^2 \, M_{34} = 2.1,
    		\\ 
    		(-1)^3 \, M_{123} =28, \quad (-1)^3 \, M_{124} = 93.9, \quad (-1)^3 \, M_{134} = 44.9, \quad (-1)^3 \, M_{234} = 12.1,
    		\\ 
    		(-1)^4 M_{1234} = 344.8.
    	\end{align*}
    	Thus, Theorem \ref{th:noturing} implies that the matrix $\jac \mbf f(\mbf P)$ will not be able to have a zero eigenvalue when adding the diffusion terms. Nevertheless, note that $S_{234}$ has the  following eigenvalues:
    	\begin{align*}
    		\lambda_{1, 2} = 0.058337 \pm 2.33564 i, \quad \lambda_3 = -2.21667.
    	\end{align*}
    	This implies, when choosing the diffusion rates $D_2 = D_3 = D_4 = 0$, that there will exist a value of $\mu > 0$ in which the Jacobian matrix $\jac \mbf f(\mbf P) - \mu \, \mathbb D$ goes unstable. Thus, $\jac \mbf f(\mbf P)$ presents a counterexample to the conjecture stated in \cite[p.144]{wang2001} on necessary and sufficient conditions for stability.
    \end{remark}
    \begin{remark}\label{differentbehaviors}
    	Theorem \ref{biggestresult} provides information about the kind of transitions we can find for the eigenvalues of $\jac \mbf f(\mbf 0) - \mu \, \mathbb D$ from having negative to positive real part, as $\mu$ ranges from 0 to $\infty$. If the transition occurs for a real (resp. complex) \evs{eigenvalue, we} will have a Turing (resp. wave) instability. Nevertheless, \evs{lemmas} \ref{positiveeig} and \ref{nonpositiveeig} tell us that, even though these transitions could involve real (resp. complex) eigenvalues, the values to which the eigenvalues of $\jac \mbf f(\mbf P) - \mu \, \mathbb D$ converge when setting some diffusion rates as zero could be complex (resp. real). Such a dispersion relation would give rise to both kinds of instabilities for the same values of parameters and diffusion rates, but for different wavenumbers. For instance, consider the following matrix
    	\begin{align}
    		\jac \mbf f(\mbf 0) &= \begin{pmatrix}
    			-21.3 & -78.2 & 87
    			\\
    			47.4 & -1 & 1
    			\\
    			-35 & -3.1 & 3
    		\end{pmatrix}, \label{wonderfulDU}
    	\end{align}
        which has eigenvalues $\lambda_\pm = -9.50278 \pm 81.3221 i$ and  $\lambda_3 = -0.29445$. Moreover, $S_{23}$ has eigenvalues $\lambda_1 = 1.94868$ and $\lambda_2 = 0.0513167$. This implies that, when choosing $D_2 = D_3 = 0$, two eigenvalues of $\jac \mbf f(\mbf 0) - \mu \, \mathbb D$ will tend to $\lambda_1$ and $\lambda_2$ as $\mu \to \infty$. However, according to Proposition \ref{th:hopf3}, the transition of the eigenvalues from having a negative to a positive real part occurs in a complex way. Figure \ref{fig:wonderful_example} shows what happens to the two largest-real-part eigenvalues of $\jac \mbf f(\mbf 0) - k^2 \, \mathbb D$ in this case. In particular, the crossing with the $k$ axis occurs at $k\approx 18.3146$, where the eigenvalues are $\lambda_\pm = \pm 0.827364 i$, and $\lambda_3 = -3373.55$. Nevertheless, there is a value of $k$ for which two real eigenvalues get separated from this single line. This implies that, when increasing the zero diffusion rates a bit, by continuity, we will have two kinds of diffusion-driven instabilities for different wavenumbers.
    \end{remark}
    \begin{figure}
    	\centering
    	\includegraphics[width=0.6\textwidth]{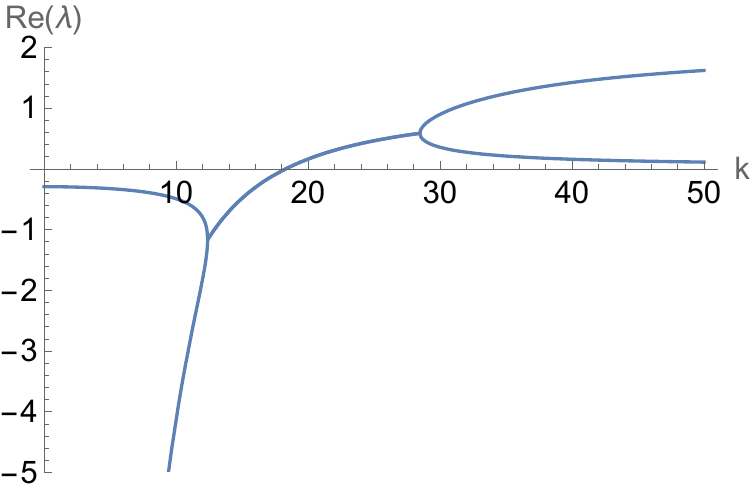}
    	\caption{Eigenvalues with the largest real parts of the matrix $\jac \mbf f(\mbf P) - k^2 \, \mathbb D$, where $\jac \mbf f(\mbf P)$ is given by \eqref{wonderfulDU}, and $\mathbb D = \diag\left(10, 0, 0\right)$.}
    	\label{fig:wonderful_example}
    \end{figure}
    
    Last but not least, until now we have only considered {\em instabilities}, corresponding to the existence of a zero-real-part eigenvalue, and not {\em bifurcation points} at which a dispersion curve first crosses the $k$-axis. Our final results give conditions under which bifurcations occur under the variation of parameters (including diffusion constants). We start with a lemma.
    \begin{lemma}
    	Let $\mbf P$ be a stable homogeneous steady state of \eqref{generalsystem} in the absence of diffusion. If for some diagonal matrix $\mathbb D$ with non-negative entries and $\mu > 0$, $\jac \mbf f(\mbf P) - \mu \, \mathbb D$ has at least one eigenvalue with a positive real part for some $\mu = \mu^*$, then there exist a diffusion matrix $\mathbb D^*$ with positive diffusion rates such that $\mbf P$ goes through a Turing or wave bifurcation under a change of parameters that takes $\mathbb D$ to $\mathbb D^*$.
    \end{lemma}
    \begin{proof}
    	Thanks to Lemma \ref{nonpositiveeig}, we know that if all the diffusion rates tend to $\infty$, then the real parts of the eigenvalues of $\jac \mbf f(\mbf P) - \mu \, \mathbb D$ tend to $-\infty$ for all $\mu > 0$. Therefore, as there exists a matrix $D$ with non-negative entries and $\mu > 0$ such that $\jac \mbf f(\mbf P)$ has an eigenvalue with a positive real part, then by continuity, if we increase all the diffusion rates sufficiently enough, we will be able to find a matrix $D^*$ such that $\jac \mbf f(\mbf P) - \mu \, \mathbb D^*$ has one eigenvalue that is tangent to the axis $\Re(\lambda) = 0$ for some $\mu > 0$, which will give rise to a bifurcation.
    \end{proof}
    In particular, this lemma is helpful when proving one last Theorem:
    \begin{theorem} \label{thm:bifurcations}
    	Suppose that the function $\mbf f$ in \eqref{generalsystem} depends on a parameter $\varepsilon\in \mathbb R$. Also, suppose that $\mbf P$ is a stable homogeneous steady state for all $\varepsilon$ sufficiently small. Assume that there exists only one increasing word $\left(i_1, i_2, \ldots, i_k\right)$ such that $S_{i_1i_2\ldots i_k}$ has a single real eigenvalue (respectively, a	purely imaginary pair of eigenvalues) that crosses the imaginary axis as $\varepsilon$ increases through zero. Then, for $\varepsilon > 0$ sufficiently small, there exists a positive diffusion matrix $\mathbb D^*$ such that $\mbf P$ goes through a Turing (respectively, wave) bifurcation under small variation as diffusion constants are generically varied from their values to $\mathbb D^*$.
    \end{theorem}
    \begin{proof}
    	The proof is a straightforward consequence of the ideas we have developed so far. First, let us choose $D_{i_1} = D_{i_2} = \ldots = D_{i_k} = 0$, while leaving all the other diffusion rates positive. This implies, thanks to \evs{Lemma} \ref{wonderfulemma}, that the eigenvalues of $\jac \mbf f(\mbf P) - \mu \, \mathbb D$ will converge to numbers with a negative real part when $\varepsilon < 0$, while some of them will converge to numbers with a positive real part when $\varepsilon > 0$. Therefore, thanks to the previous lemma, we can increase the zero diffusion rates such that we find a bifurcation. The type of bifurcation is characterized by continuity. In fact, for a sufficiently small value of $\varepsilon$, one or two eigenvalues of $\jac \mbf f(\mbf P) - \mu \, \mathbb D$ will be converging to one or two eigenvalues with a positive real part, so increasing the diffusion rates just a little will not change the type of eigenvalue we are having as $\mu\to \infty$.
    \end{proof}
    \begin{remark}        
    	Note that, owing to linearity and the last part of Lemma \ref{signpreservation}, only the ratios between the diffusion rates matter. The wavenumber $k$ that corresponds to an instability can be made arbitrarily small or large, by scaling all diffusion constants by the same factor $\delta > 0$.
    \end{remark}
    
    To close the section, we can state the following result, which is a straightforward consequence of \evs{Theorems} \ref{th:sat} and \ref{biggestresult}.
    \begin{theorem}
        Let $\mbf P$ be a stable homogeneous steady state of \eqref{generalsystem}. Then linear diffusion-driven instabilities can occur around $\mbf P$ if and only if the system is stable about $\mbf P$ in the absence of diffusion but there is a principal unstable subsystem. In particular,
        \begin{itemize}
            \item if $\jac \mbf f(\mbf P)$ is s-stable, then the system cannot show a Turing or wave instability around $\mbf P$; or
            \item if a principal submatrix of $\jac \mbf f(\mbf P)$ is unstable, then there exist non-negative diffusion rates such that the full system shows a Turing or wave instability around $\mbf P$.
        \end{itemize}
        Furthermore, in the latter case, the diffusion rates can be chosen as stated in Theorem \ref{biggestresult}.
    \end{theorem}


\section{\evs{Applications}} \label{sec:examples}
    We consider \evs{several} different {examples to illustrate the theory, in each case a model for a biophysical process that has arisen in the literature.}
    
    \evs{The final three examples each illustrate how Theorem \ref{biggestresult} can be used to classify the crossings of dispersion curves with the axis $\Re(\lambda) = 0$ as being either real or complex.} Then, Lemmas \ref{positiveeig} and \ref{nonpositiveeig} \evs{give} an idea about the behaviour of the eigenvalues of $\jac \mbf f(\mbf P) - \mu \, \mathbb D$ as $\mu \to \infty$, when some diffusion rates are large or equal to zero. To see the importance of those results, see Remark \ref{differentbehaviors}. In particular, we highlight that there are two cases that generate wave instabilities automatically after choosing the diffusion rates in a sensible way. These are: \evs{First}, one of the principal submatrices of $\jac \mbf f(\mbf P)$ has a complex-conjugate pair of eigenvalues with positive real part, which generates wave instabilities for large wavenumbers, or \evs{second}, the number of eigenvalues with a positive real part grows faster than the number of changes of sign in the terms $(-1)^k \, M_{s(j_1, j_2, \ldots, j_k)}$ stated in Theorem \ref{biggestresult}.

    \subsection{Lack of instability in a model for host-parasyte interaction}
        \evs{First though, we provide an example where our theory can show the non-existence of instability. We consider a realistic model for the spread of malaria}, based on the classic Ross-MacDonald model (see \cite{alonso}, and references therein). This model treats the dynamics of the infection between humans and mosquitoes. In particular, let $H$ (resp. $I$) be the population density of healthy (resp. infected) people. Moreover, let $P$ be the density of the population of infected mosquitoes. We set the following system to model the dispersion of malaria:
        \begin{align*}
        	\partial_t H &= \left(b_H - d_H\right) H - c \, P \, H + r \, I + D_1 \, \partial_{xx} H,
        	\\
        	\partial_t I &= - d_H \, I + c \, P \, H - r \, I + D_2 \, \partial_{xx} I,
        	\\
        	\partial_t P &= - d_M \, P + b \, (Q - P) \, I + D_3 \, \partial_{xx} P.
        \end{align*}
        Note that healthy mosquitoes do not appear in this model. Nevertheless, they are considered under the assumption that $Q$ is the total mosquito population so that $Q - P$ is the population of healthy mosquitoes.
        
        To preserve the significance of this model, we explicitly state that all the variables and parameters are non-negative, $b_H > d_H > 0$ and $P \leq Q$. This system has only two equilibria, $\mbf 0 = (0, 0, 0)$ and $\mbf P = (H^*, I^*, P^*)$, where
        \begin{align*}
        	H^* &= \frac{d_H \, d_M \left(d_H + r\right)}{b \left(d_H \left(c \, Q + d_H + r\right) - b_H \left(d_H + r\right)\right)},
        	\\
        	I^* &= \frac{d_M \left(b_H - d_H\right) \left(d_H + r\right)}{b \left(d_H \left(c \, Q + d_H + r\right) - b_H \left(d_H + r\right)\right)},
        	\\
        	P^* &= \frac{\left(b_H - d_H\right) \left(d_H + r\right)}{c \, d_H}.
        \end{align*}
        Here, \evs{we assume $d_H \left(c \, Q + d_H + r\right) - b_H \left(d_H + r\right) > 0$}. Note that a Turing instability pattern around the origin will not be biologically relevant since the existence of it would imply that some of the variables become negative. On the other \evs{hand, the Jacobian} matrix of the system at $\mbf P$ in the absence of diffusion is given by
        \begin{align*}
        	\jac \mbf f(\mbf P) = \begin{pmatrix}
        		r - \frac{r \, b_H}{d_H} & r & a_{13}
        		\\[1ex]
        		\frac{\left(b_H - d_H\right) \left(d_H + r\right)}{d_H} & - (d_H + r) & a_{23}
        		\\[1ex]
        		0 & b \left(Q - \frac{\left(b_H - d_H\right) \left(d_H + r\right)}{c \, d_H}\right) & M_3
        	\end{pmatrix},
        \end{align*}
        where
        \begin{align*}
        	a_{13} &= - \frac{c \, d_H \, d_M \left(d_H + r\right)}{b \left(d_H \left(c \, Q + d_H + r\right) - b_H \left(d_H + r\right)\right)},
        	\\
        	a_{23} &= \frac{c \, d_H \, d_M \left(d_H + r\right)}{b \left(d_H \left(c \, Q + d_H + r\right) - b_H \left(d_H + r\right)\right)},
        	\\
        	M_3 &= - \frac{c \, Q \, d_H \, d_M}{d_H \left(c \, Q + d_H + r\right) - b_H \left(d_H + r\right)}.
        \end{align*}
        Now, note that $M_1, M_2, M_3\leq 0$. Moreover, observe that $M_{12}=0$, 
        \begin{align*}
        	M_{13} &= \frac{c \, Q \, r \, d_M \left(b_H - d_H\right)}{d_H \left(c \, Q + d_H + r\right) - b_H \left(d_H + r\right)}>0,
        	\\
        	M_{23} &= \frac{d_M \left(b_H - d_H\right) \left(d_H + r\right)^2}{d_H \left(c \, Q + d_H + r\right) - b_H \left(d_H + r\right)} > 0,
        	\\
        	M_{123} &= - \left(d_M \left(b_H - d_H\right) \left(d_H + r\right)\right) < 0.
        \end{align*}
        This implies, according to Lemma \ref{hopf3sufficient} and Theorem \ref{th:noturing}, that this system will not have Turing or wave instabilities for any parameter values.
      
    \subsection{A model for wavetrains in excitable media}
        As a next example, we consider the model proposed by Yochelis {\em et al} \cite{Yochelis2008} based on a generalisation of the FitzHugh-Nagumo system. The system can be written as
        \begin{align}
        	\frac{\partial u}{\partial t} &= u - u^3 - v + D_1 \, \nabla^2 u\evs{,} \nonumber
        	\\
        	\frac{\partial v}{\partial t} &= \epsilon_v \left(u - a_v \, v - a_w \, w - a_0\right) + D_2 \, \nabla^2 v, \label{example}
        	\\
        	\frac{\partial w}{\partial t} &= \epsilon_w \left(u - w\right) + D_3 \, \nabla^2 w, \nonumber
        \end{align}
        where $a_v, a_w, \epsilon_v, \epsilon_w > 0$, and $D_1, D_2, D_3\geq 0$ are parameters. See \cite{Yochelis2008} for a detailed interpretation of each of the variables and parameters. In that paper, the authors found \evs{numerical} evidence for both Turing and wave bifurcations for particular values of the parameters. 
        
        System \eqref{example} has, at most, three homogeneous steady states. We shall focus on the one given by $\mbf P := \left(u^*, v^*, w^*\right)$, where
        \begin{align*}
        	u^* &= - \frac{6 \, a_v^2 + 6 \, a_v (a_w - 1) + \sqrt[3]{2} \, Q^2}{3 \cdot 2^{2/3} \, a_v \, Q}, \quad w^* = - \frac{6 \, a_v^2 + 6 \, a_v (a_w - 1) + \sqrt[3]{2} \, Q^2}{3 \cdot 2^{2/3} \, a_v \, Q},
        	\\ 
        	v^* &= \frac{9 \, \left(a_w - 1\right) \left(6 \sqrt[3]{2} \, a_v^3 + 6 \sqrt[3]{2} \, a_v^2 (a_w - 1) + 2^{2/3} \, a_v \, Q^2 \right) - 2 \, \sqrt{L} \, Q + 2 \, Q^4}{54 \, a_v^3 \, Q},
        	\\
        	& \mbox{with} \qquad Q = \left(\sqrt{L} - 27 \, a_0 \, a_v^2\right)^\frac{1}{3}, \quad  L = 729 \, a_0^2 \, a_v^4 - 108 \, a_v^3 (a_v + a_w - 1)^3,
        \end{align*}
        which exists provided $L>0$. The Jacobian matrix of \eqref{example}, evaluated at $\mbf P$, can be readily computed to be 
        \begin{align*}
        	\jac \mbf f(\mbf P) &= \begin{pmatrix}
        		1 - R^2 & -1 & 0
        		\\ 
        		\epsilon_v & - a_v \, \epsilon_v & - a_w \, \epsilon_v
        		\\ 
        		\epsilon_w & 0 & - \epsilon_w
        	\end{pmatrix}\evs{,} \qquad \mbox{where}
            \\
            R^2 &= \frac{\left(6 \, a_v^2 + 6 \, a_v \left(a_w - 1\right) + \sqrt[3]{2} \, Q^2\right)^2}{6 \sqrt[3]{2} \, a_v^2 \, Q^2},
        \end{align*}
        from which the principal minors can be easily computed
        \begin{align}
        	M_1 &= 1 - R^2, \quad  M_2 = - a_v \, \epsilon_v, \quad M_3 = - \epsilon_w, \quad M_{12} = \epsilon_v \left(1 + R^2 - a_v\right), \nonumber 
        	\\ 
        	M_{13} &= \epsilon_w \left(R^2 - 1\right), \quad  M_{23} = a_v \, \epsilon_v \, \epsilon_w, \quad M_{123} = \epsilon_v \, \epsilon_w (a_v \left(1 - R^2\right) + q_w - 1). \label{eq:M_Arik}
        \end{align}
        Note that $M_2,M_3<0$, $M_{23}>0$, whereas $M_1$, $M_{12}$ $M_{13}$ and $M_{123}$ can take either sign, depending on the value of $R^2$.
        
        From Lemma \ref{necessarystability}, \evs{the solution $\mbf P$ is stable if the following conditions are satisfied}
        \begin{align}
        	M_{123} < 0, \qquad M_1 + M_2 + M_3 < 0, \notag
            \\
            (M_1 + M_2 + M_3)(M_{12} + M_{13} + M_{23})< M_{123}\evs{,} \label{eq:P_Arik}
        \end{align}
        which lead to an open set that can readily be written in terms of the problem parameters using \eqref{eq:M_Arik}. 
        
        From here on, we will assume that conditions \eqref{eq:P_Arik} hold so that $\mbf P$ is a stable equilibrium in the absence of diffusion. Now, using the results proven in the previous sections, we can state the following.
        \begin{proposition}
        	$M_1 > 0$ if and only if there exist non-negative diffusion rates $D_1, D_2, D_3$ such that \eqref{example} shows a Turing instability around $\mbf P$.
        \end{proposition}
        \begin{proof}
        	Note from \eqref{eq:M_Arik} that $M_{13} = -\epsilon_w \, M_1$, so that $M_1 > 0$ implies $M_{13} < 0$. Note also that the sequence $(-1) \, M_1$, $(-1)^2 \, M_{13}$, $(-1)^3 \, M_{132}$ changes its sign only once. Furthermore, as $M_{13} < 0$, then we know that $S_{13}$ has only one positive eigenvalue. Thus, according to Theorem \ref{biggestresult}, we can find non-negative diffusion rates such that \eqref{example} admits a Turing instability around $\mbf P$.
        	
        	On the other hand, if $M_1\leq 0$, then $M_{12}\geq 0$ and $M_{13}\geq 0$. This implies that $(-1)^3 \, M_{123}$, $(-1)^2 \, M_{ij}$ and $(-1) \, M_i$ are all strictly non-negative for each $1 < i < j < 3$. Then, Lemma \ref{th:noturing} implies that \eqref{example} will not be able to develop a Turing instability around $\mbf P$ for any positive diffusion constants.
        \end{proof}
        \begin{proposition}
        	If $M_1 + M_2 > 0$ and $M_{12} > 0$, then we can choose non-negative diffusion rates $D_1, D_2, D_3$ such that \eqref{example} admits a wave instability around $\mbf P$.
        \end{proposition}
        \begin{proof}
        	If $M_1 + M_2 > 0$ and $M_{12} > 0$, then $S_{12}$ will have two eigenvalues with a positive real part. Furthermore, we will have no change of sign from $- M_{123}$ to $M_{12}$. Then, according to Theorem \ref{biggestresult}, we will be able to find non-negative diffusion rates such that \eqref{example} admits a wave instability around $\mbf P$.
        \end{proof}
        For a numerical illustration of these results, we shall consider the following parameter values extracted from \cite{Yochelis2008}:
        \begin{align}
        	a_v = 0.435, \quad a_w = 0.5, \quad  \epsilon_v = 0.2 \quad \epsilon_w = 1, \quad \mbox{and} \quad a_0 = -0.1.  \label{pars_Arik}
        \end{align}
        which can easily be seen to lie inside the parameter region in which $\mbf P$ is stable. Furthermore, at these parameter values we have $M_1 = 0.151497 > 0$, $M_2 = -0.087 < 0$, $M_3 = -1 < 0$, $M_{12} = 0.18682 > 0$, $M_{13} = -0.151497 < 0$ $M_{23} = 0.087 > 0$ and  $M_{123} = -0.0868198 < 0$, from which we see that $M_1 + M_2 > 0$. Thus, the above two propositions show that there must exist a set of non-negative diffusion rates such that $\mbf P$ undergoes a Turing instability and another set under which it undergoes a wave instability. The eigenvalues of $S_{12}$ are given by $0.0322485 \pm 0.431022 i$, which have a positive real part, giving conditions for a wave instability. Moreover, $S_{13}$ has eigenvalues given by $-1$ and $0.151497$ which will give rise to a Turing instability.
        
        Consider first the wave bifurcation. Let us define a parameter $\delta > 0$ and consider $D_1 = D_2 = \delta$ and $D_3 = 0.1$. The corresponding dispersion curves are plotted in Fig.~\ref{fig:dispersionrelationTH}(a) for $\delta$ ranging from $0$ to $0.01$ in uniform steps. Note that when $\delta = 0$ one of the dispersion relation curves does indeed tend to a finite positive value ($\approx 0.032249$), as $\mu\to \infty$, exactly as shown by Lemma \ref{positiveeig}. \evs{Therefore, for $\delta > 0$ sufficiently small,} there must \evs{exist} an interval of values of $k$ for which the dispersion curve is positive. As $\delta$ is increased through a critical value $\delta^* \approx 2.27312 \times 10^{-3}$, there is a point of tangency between this dispersion curve and the $k$-axis, at a critical wavenumber $k^* \approx 2.43287$. To check that this is indeed a wave zero, we have computed the critical eigenvalues of $\jac \mbf f(\mbf P) - \mu \, \mathbb D$ to be $-1.55429$ and $\pm 0.355382 i$. Furthermore, Fig.~\ref{fig:dispersionrelationTH}(b) shows the results of a numerical computation in one spatial dimension, which was performed at a $\delta$-value just beyond that at which the bifurcation occurs, which indeed confirms the onset of a spatiotemporal instability.
        
        \begin{figure}
        	\centering
        	\includegraphics[width=0.45\textwidth]{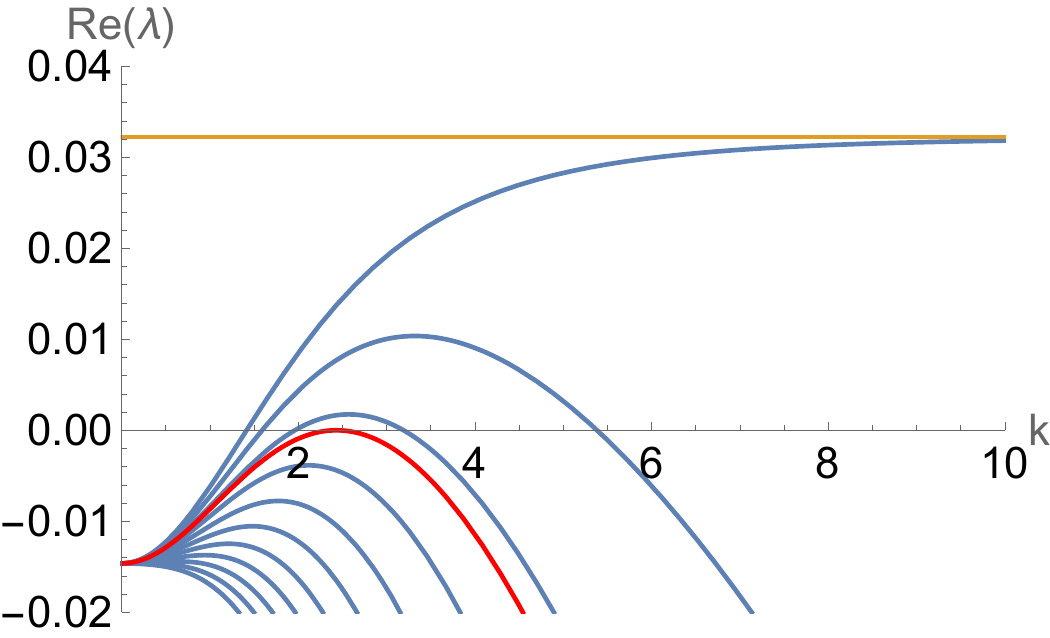} \ 
        	\includegraphics[width=0.45\textwidth]{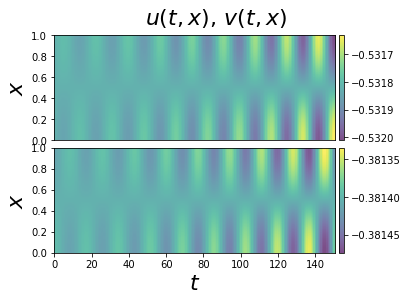}
        	\caption{(a) Dispersion relation of model \eqref{example} at parameter values \eqref{pars_Arik} $\left(a_v, a_w, \epsilon_v, \epsilon_w, a_0\right) = (0.435, 0.5, 0.2, 1, -0.1)$ and $\left(D_1, D_2, D_3\right) = (\delta, \delta, 0.1)$ as $\delta$ varies from $0$ to $0.01$ in uniform steps. The red line is the dispersion curve for $\delta = 2.27312\times 10^{-3}$, where the wave bifurcation occurs, and the orange line is the straight line $\Re(\lambda) = 0.0322485$. (b) Simulation of \eqref{example} on a 1D domain of length $L = 1$ with Neumann boundary conditions, for $\delta = 0.001$ and an initial condition that is a perturbation of $\mbf P$ proportional to $\cos(\pi x)$ with amplitude $10^{-5}$ in all three components.}
        	\label{fig:dispersionrelationTH}
        \end{figure}
        To find a Turing instability, let us introduce diffusion rates $D_2 = 5$, while $D_1 = D_3 = \delta$, where again, $\delta$ is a parameter. Figure \ref{fig:dispersionrelationturing}(a) shows dispersion curves for this case as $\delta$ ranges from $0$ to $0.5$ in uniform steps. In addition, the curve for $\delta\approx 0.28684$ shows the bifurcation parameter value, for which the critical wavenumber is $k \approx 0.496617$, and the eigenvalues of $\jac \mbf f(\mbf P) - \mu \, \mathbb D$ are $0$ and $-1.15507 \pm 0.293739  i$. Thus, we have a Turing bifurcation for these parameter values.
        \begin{figure}[tbp]
        	\begin{subfigure}[b]{0.45\textwidth}
        		\centering
        		\includegraphics[width=\textwidth]{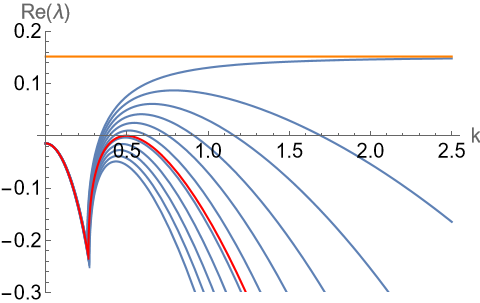}
        		\caption{ }
        	\end{subfigure}
        	\hfill
        	\begin{subfigure}[b]{0.45\textwidth}
        		\centering
        		\includegraphics[width=\textwidth]{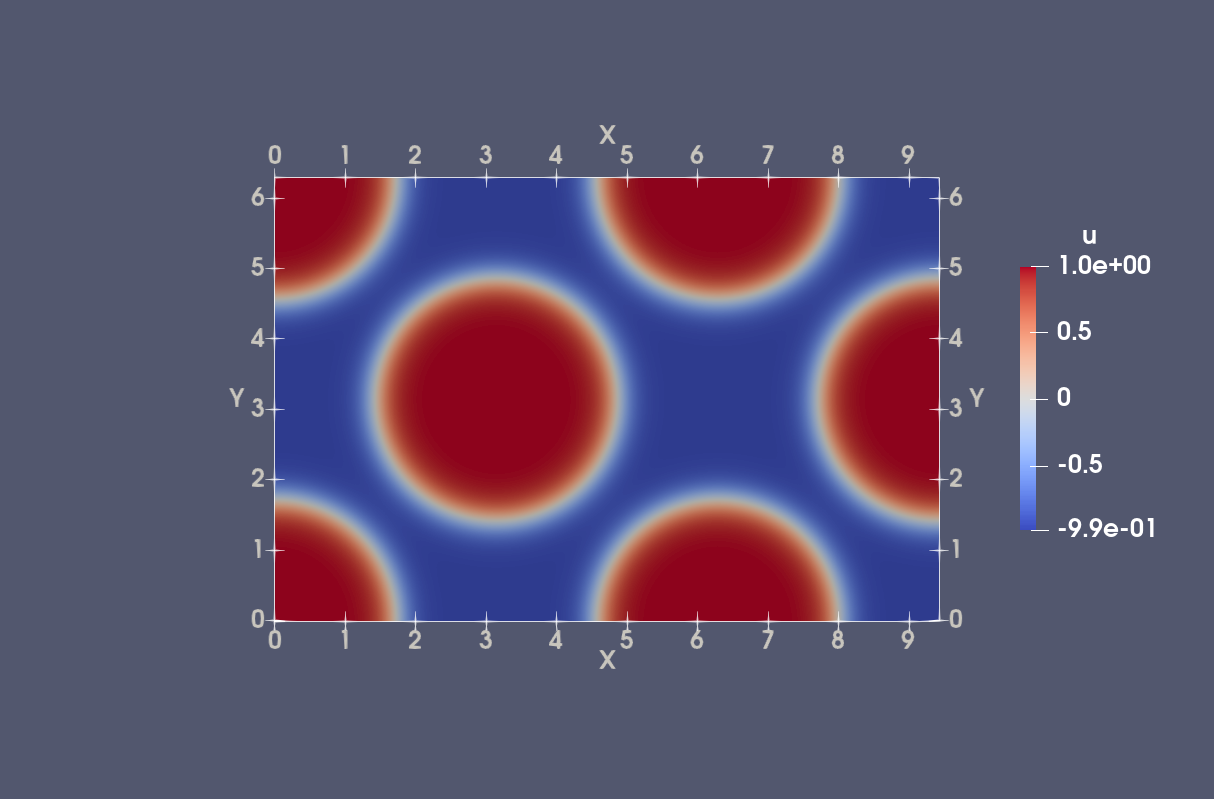}
        		\caption{ }
        	\end{subfigure}
        	\\
        	\begin{subfigure}[b]{0.45\textwidth}
        		\centering
        		\includegraphics[width=\textwidth]{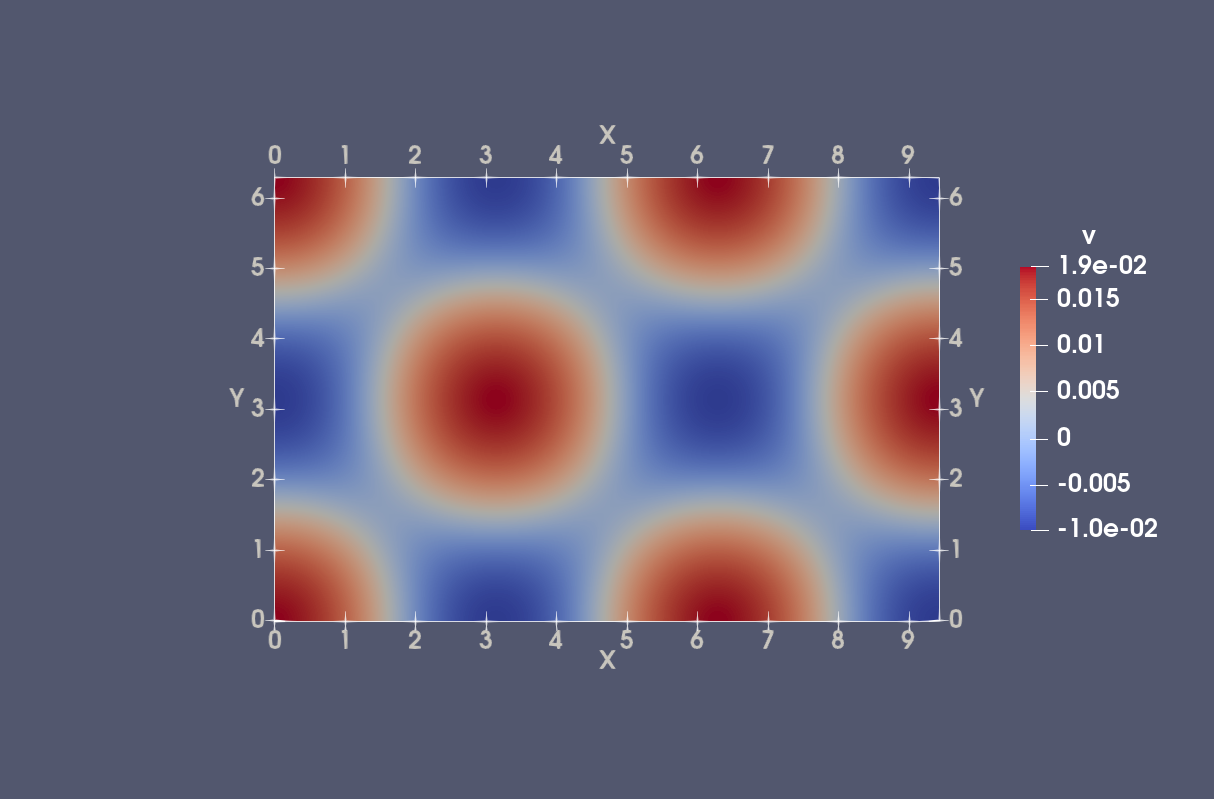}
        		\caption{ }
        	\end{subfigure}
        	\hfill
        	\begin{subfigure}[b]{0.45\textwidth}
        		\centering
        		\includegraphics[width=\textwidth]{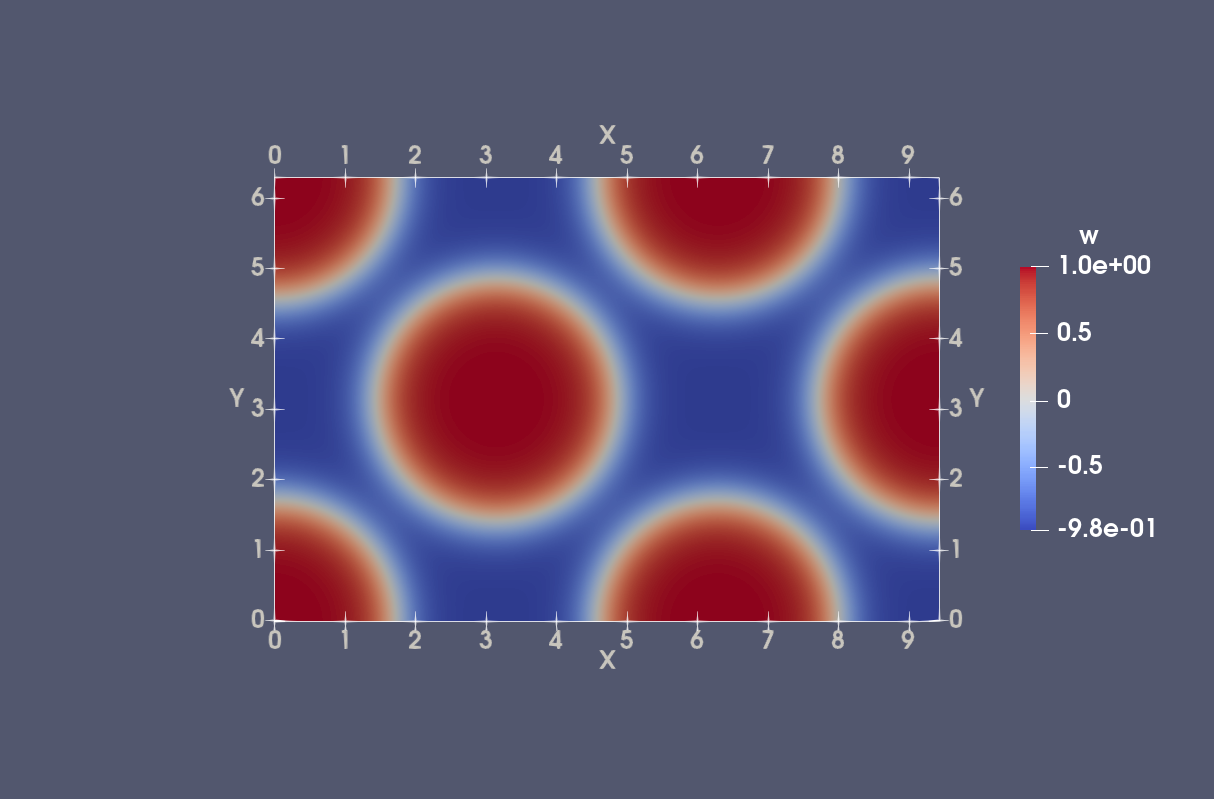}
        		\caption{ }
        	\end{subfigure}
        	\caption{(a) Similar to Fig.~\ref{fig:dispersionrelationTH}(a) but for $\left(D_1, D_2, D_3\right) = (\delta, 5, \delta)$ and $\delta$ varying from 0 to 0.5 in uniform steps; the distinguished (red) curve is for $\delta = 0.28684$, where a Turing bifurcation occurs, and the (orange) straight line is at $\Re(\lambda) = 0.151497$. (b), (c), (d) variables $u$, $v$, $w$, respectively, at the Turing instability pattern obtained after integrating the system in two spatial dimensions with $\delta = 0.03$. Here, we make a perturbation in all components equally corresponding to wave numbers $(k_x, k_y) = (3, 2)$ with an amplitude $\varepsilon = 10^{-4}$.}
        	\label{fig:dispersionrelationturing}
        \end{figure}
        This result is illustrated in Fig.~\ref{fig:dispersionrelationturing} where we show both the details of the dispersion relation in this case and also a computation in two spatial dimensions.
    
    \subsection{Symmetrically coupled Brusselator systems}
        In this section, we consider a model of two symmetrically coupled Brusselator systems studied by Konishi and Hara \evs{on a finite domain} in \cite{Konishi2018}, given by 
        \begin{align}
        	\frac{\partial u}{\partial t} &= \eta \left(a - (b + 1) \, u + u^2 \, v\right) + c \, (w - u) + D_1 \, u_{xx}, \nonumber
        	\\[1.5ex]
        	\frac{\partial v}{\partial t} &= \eta\left(- u^2 \, v + b \, u\right) + c \, (z - v) + D_2 \, v_{xx}, \label{Bru4}
        	\\[1.5ex]
        	\frac{\partial w}{\partial t} &= \eta\left(a - (b + 1) \, w + w^2 \, z\right) + c \, (u - w) + D_3 \, w_{xx}, \nonumber
        	\\ 
        	\frac{\partial z}{\partial t} &= \eta\left(- w^2 \, z + b \, w\right) + c \, (v - z) + D_4 \, z_{xx}, \nonumber 
        \end{align}
        Here, $u(x, t)$, $v(x, t)$, $w(x, t)$, and $z(x, t)$ are scalar fields, $a, b, c, \eta > 0$ are reaction parameters while $D_1, D_2, D_3, D_4\geq 0$ are diffusion rates. In \cite{Konishi2018} it is shown numerically that the symmetrically coupled system may have two distinct Turing bifurcations. We show here how this fits into our theory. 
        
        The model has at most five homogeneous steady states. We are interested in the simplest one 
        \begin{align}
        	\mbf P = \left(a, \frac{b}{a}, a, \frac{b}{a}\right), \label{eq:P_bruss}
        \end{align}
        the Jacobian matrix around which, in the absence of diffusion, is given by 
        \begin{align}
        	\jac \mbf f(\mbf P) = \begin{pmatrix}
        		(b - 1) \, \eta - c & a^2 \, \eta  & c & 0
        		\\
        		- b \, \eta  & - a^2 \, \eta - c & 0 & c
        		\\
        		c & 0 & (b - 1) \, \eta - c & a^2 \, \eta 
        		\\
        		0 & c & - b \, \eta  & - a^2 \, \eta - c
        	\end{pmatrix}. \label{eq:Brus_Jac}
        \end{align}
        whose eigenvalues are given by
        \begin{align*}
        	\lambda_\pm^{[1]} &= \frac{\eta}{2} \left(- a^2 + b - 4 \, \frac{c}{\eta} - 1 \pm \sqrt{\left(- a^2 + b - 1\right)^2 - 4 \, a^2}\right),
        	\\ 
        	\lambda_\pm^{[2]} &= \frac{\eta}{2} \left(- a^2 + b - 1 \pm \sqrt{\left(- a^2 + b - 1\right)^2 - 4 \, a^2}\right).
        \end{align*}
        This implies that, if we have $b < a^2 + 1$, then $\mbf P$ is stable in the absence of diffusion. On the other hand, note that $M_1 = M_3 = (b - 1) \, \eta - c > 0$ if and only if $b > \frac{c}{\eta} + 1$. These inequalities immediately imply that, if $\frac{c}{\eta} + 1 < a^2 + 1$, then we can find a value of $b > 0$ such that $\frac{c}{\eta} + 1 < b < a^2 + 1$ and, according to Theorem \ref{biggestresult}, there will exist non-negative diffusion rates such that system \eqref{Bru4} \evs{exhibits} a Turing instability.
        
        Next, note that
        \begin{align*}
        	S_{13} = \begin{pmatrix}
        		(b - 1) \, \eta - c & c
        		\\ 
        		c & (b - 1) \, \eta - c
        	\end{pmatrix}.
        \end{align*}
        Therefore, $\tr\left(S_{13}\right) > 0$ if and only if $b > \frac{c}{\eta} + 1$, and  
        $\det\left(S_{13}\right) > 0$ if and only if $b > 2 \, \frac{c}{\eta} + 1$. Thus, if $2 \, \frac{c}{\eta} + 1 < a^2 + 1$, there will exist $b > 0$ such that $2 \, \frac{c}{\eta} + 1 < b < a^2 + 1$, making $S_{13}$ have two eigenvalues with a positive real part, while $\mbf P$ remains stable in the absence of diffusion.
        
        Furthermore, note that $M_{123} =  M_{134}$ and
        \begin{align*}
        	m(a) &:= (-1)^3 \, M_{123}
            \\
            &= - \eta^2 ((b - 2) \, c + (b - 1) \, \eta ) \, a^2 - (b - 1) \, c \, \eta \,  (- b \, \eta + 2 \, c + \eta)
        \end{align*}
        is a parabola in $a$ that has a negative leading coefficient, since if $b > 2 \, \frac{c}{\eta} + 1$, then $(b - 2) \, c + (b - 1) \, \eta > 2 \, \frac{c^2}{\eta} > 0$. Therefore, $m(a)$ attains its maximum value at $a = 0$, which is $- (b - 1) \, c \, \eta \, \left(- b \, \eta + 2 \, c + \eta\right) < 0$. Thus, $m(a)$ is negative definite.  With this, we have $(-1)^4 \, M_{1234} > 0$, $(-1)^3 \, M_{123} = (-1)^3 \, M_{134} < 0$, and $(-1)^2 \, M_{13} > 0$. This implies that we have two changes of sign in the characteristic polynomial equation when choosing the diffusion constants appropriately and, by Lemma \ref{positiveeig}, we must have two eigenvalues of $\jac \mbf f(\mbf P) - \mu \, \mathbb D$ with positive real part in the limit as $\mu \to \infty$.
        
        
        Figure \ref{fig:dispersionrelationsecondexampleturing} illustrates dispersion curves for each of the two Turing bifurcations. First, in panel (a), we choose diffusion rates $(D_1, D_2, D_3, D_4) = (\delta, 1.8, 0.2, 1.0)$, and let $\delta$ range from 0 to 0.6 in discrete steps of size 0.06. Furthermore, we find that the Turing bifurcation occurs for $\delta = 0.442546$.
        
        \begin{figure}
        	\centering
        	\includegraphics[width = 0.45\textwidth]{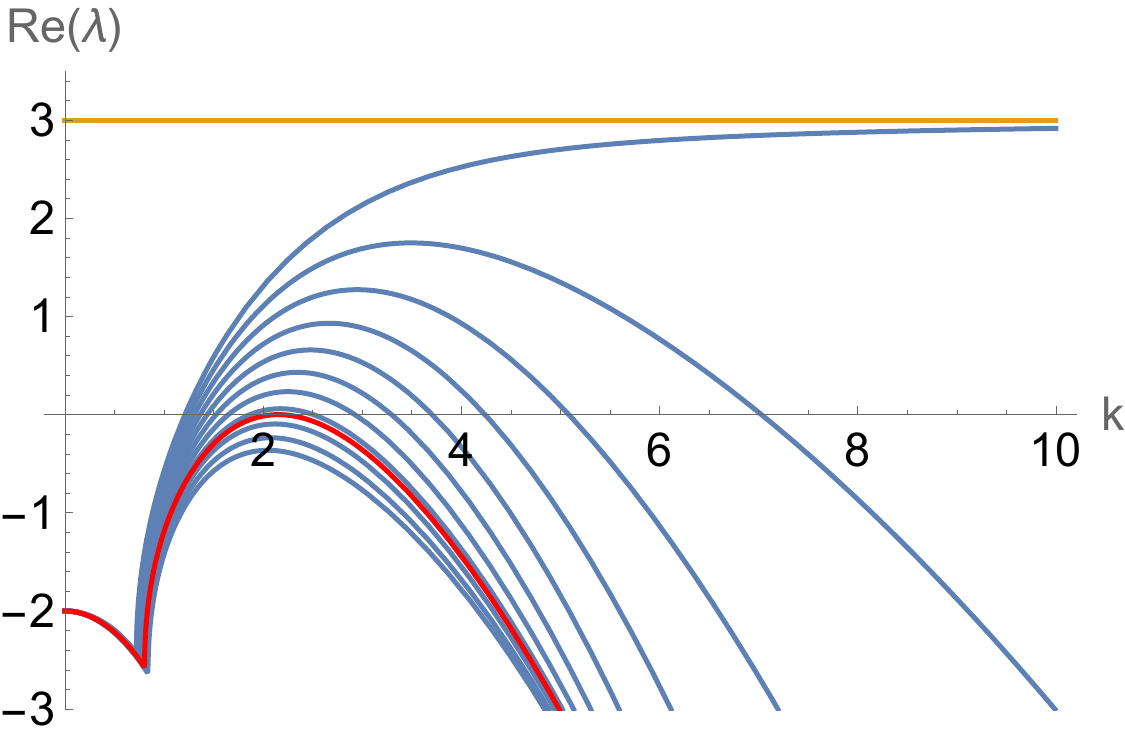}
        	\includegraphics[width = 0.45\textwidth]{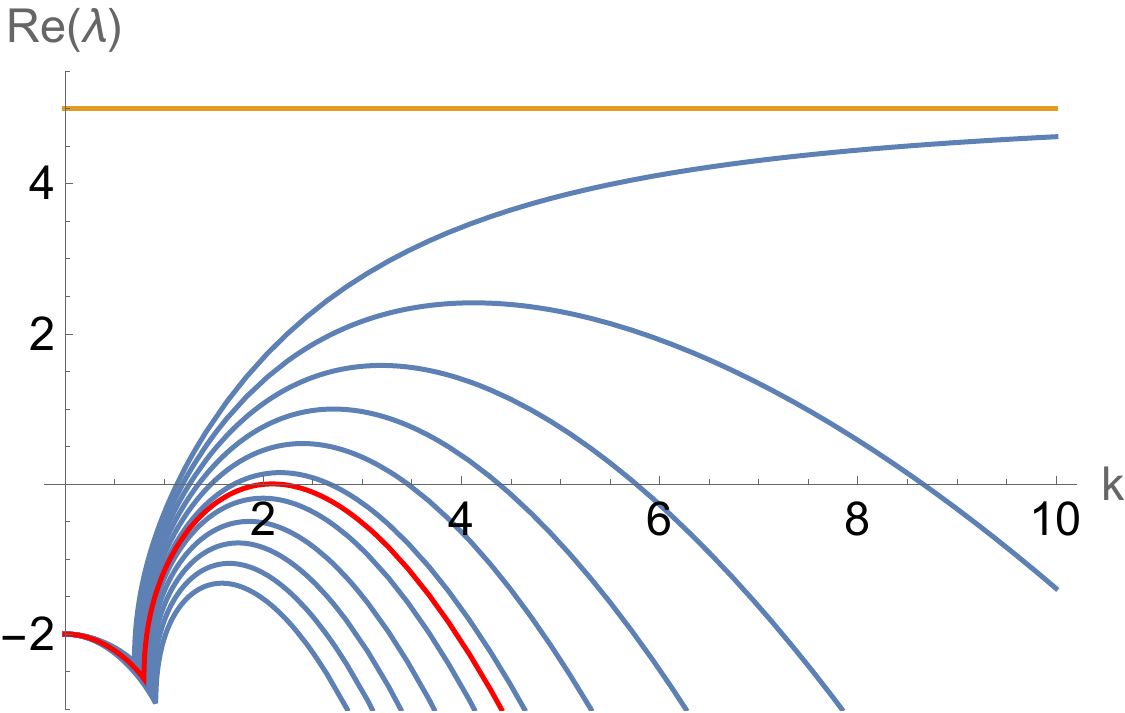}
        	\caption{(a) Dispersion relation of system \eqref{Bru4} when $(a, b, c, \eta) = (3, 6, 2, 1)$, and $\left(D_1, D_2, D_3, D_4\right) = (\delta, 1.8, 0.2, 1.0)$, as $\delta$ ranges from 0 to 0.6 in discrete steps of size 0.06. The blue lines are these dispersion relation curves, while the red line is the dispersion relation when $\delta = 0.4425461$, where the Turing bifurcation occurs, and the green line is the straight line $\Re(\lambda) = 3$. (b) Dispersion relation of system \eqref{Bru4} when $\left(a, b, c, \eta\right) = (3, 6, 2, 1)$, and $\left(D_1, D_2, D_3, D_4\right) = (\delta, 1.8, \delta, 1)$, as $\delta$ ranges from 0 to 0.6 in discrete steps of size 0.06. The blue lines are these dispersion relation curves, while the red line is the dispersion relation when $\delta = 0.3256691$, where the Turing bifurcation occurs, and the green line is the straight line $\Re\left(\lambda\right) = 5$.}
        	\label{fig:dispersionrelationsecondexampleturing}
        \end{figure}
        On the other hand, when choosing $\left(D_1, D_2, D_3, D_4\right) = (\delta, 1.8, \delta, 0.1)$, as $\delta$ ranges from 0 to 0.6, we find the dispersion curves in panel (b), for which the Turing bifurcation occurs when $\delta\approx 0.325669$.
        
        To show that these \evs{really} are independent Turing bifurcations, we can find the critical eigenvector of the operators $\jac \mbf f(\mbf 0) + \mathbb D \, \partial_{xx}$ in each case; see Fig.~\ref{fig:eigenvectors}. 
        \begin{figure}[tbp]
        	\centering
        	\begin{subfigure}[b]{0.45\textwidth}
        		\centering
        		\includegraphics[width=\textwidth]{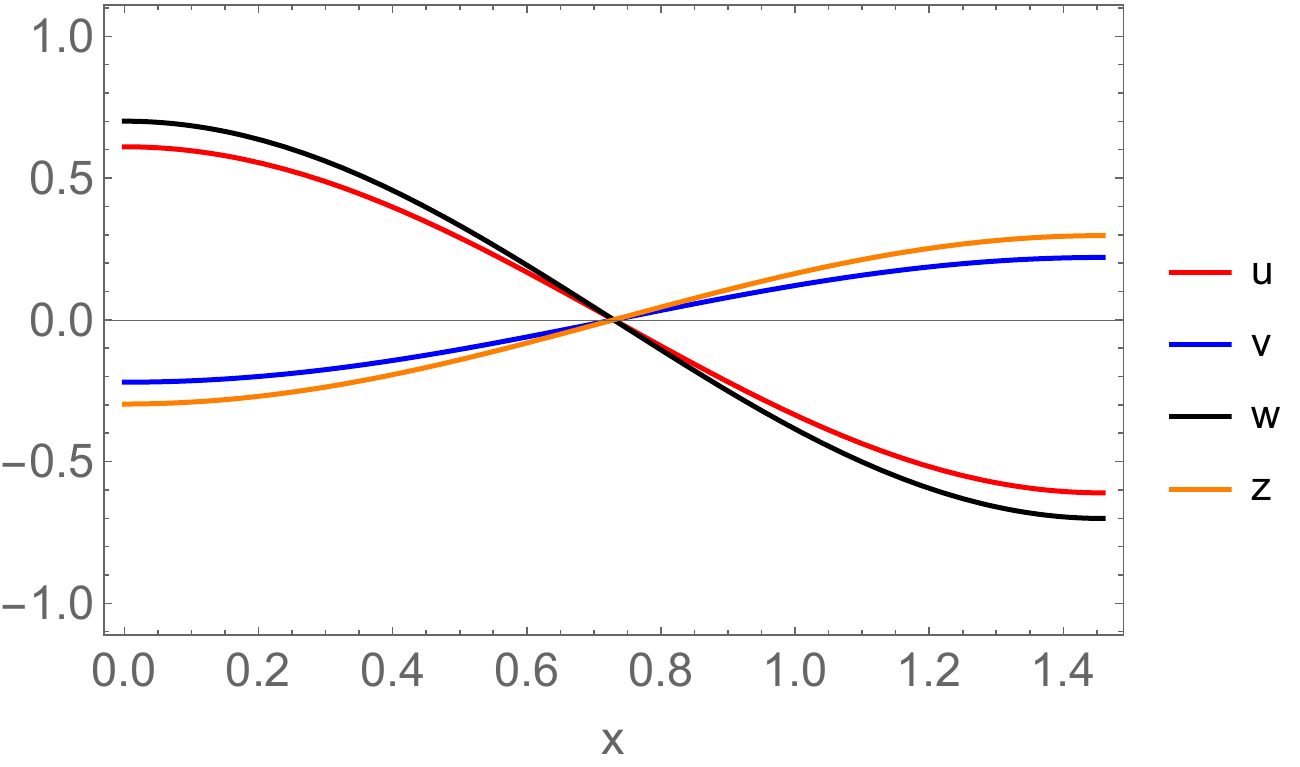}
        		\caption{}
        	\end{subfigure}
        	\hfill
        	\begin{subfigure}[b]{0.45\textwidth}
        		\centering
        		\includegraphics[width=\textwidth]{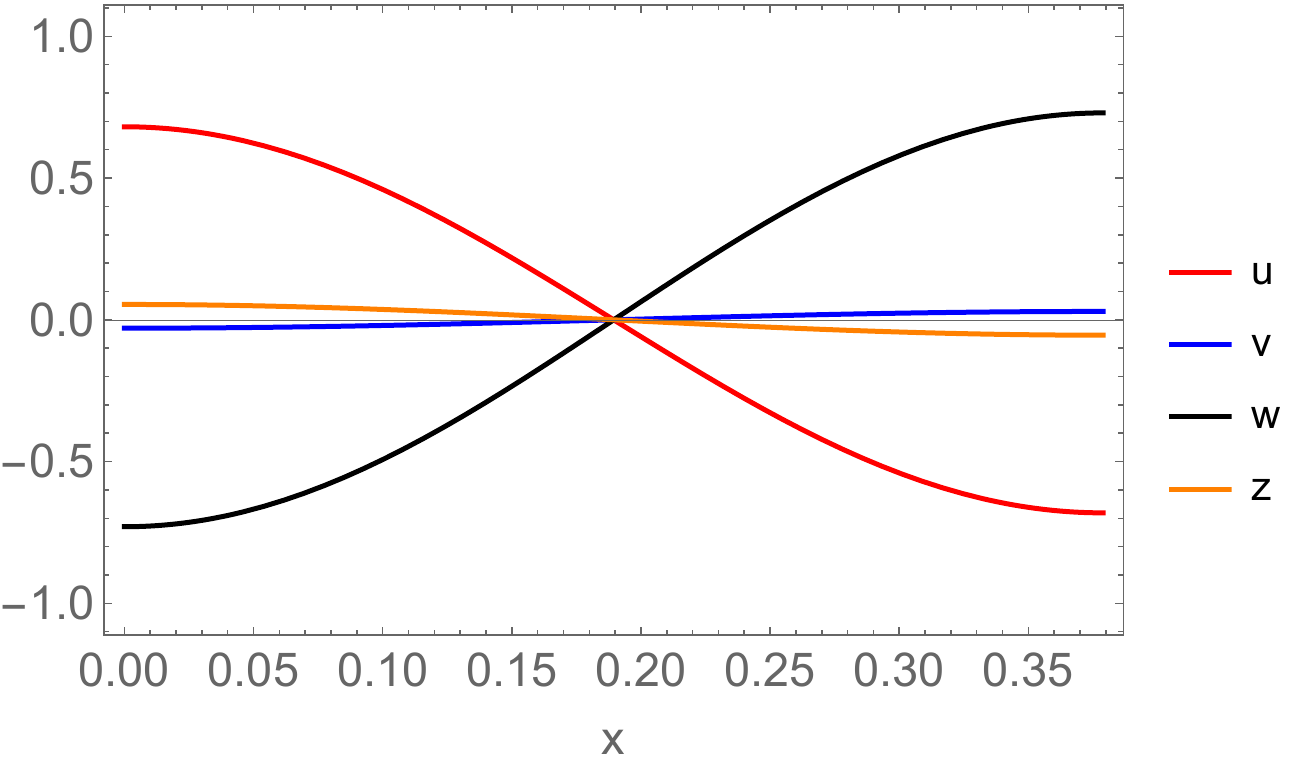}
        		\caption{}
        	\end{subfigure}
        	\caption{Graphs of a normalized eigenvector of $\jac \mbf f(\mbf P) + \mathbb D \, \partial_{xx}$. Here $\left(D_1, D_2, D_3, D_4\right)$ is equal to; (a) $(\delta, 1.8, 0.2, 1.0)$ (with $k = 2.15296$), and (b) $(\delta, 1.8, \delta, 1.0)$ (with and $k = 2.0924$)}
        	\label{fig:eigenvectors}
        \end{figure}
    
    \subsection{Asymmetrically coupled Brusselator systems}
        As a final example, we shall show how a small change to the previous example, breaking the symmetry of the coupling, can lead to the existence of a wave bifurcation. In particular, let us consider model \eqref{Bru4} with the addition of a term $- 2 \, c \left(v - \frac{b}{a}\right)$ to the right-hand side of the $z$-equation, so that it becomes: 
        \begin{align}
        	\frac{\partial z}{\partial t} = \eta\left(- w^2 \, z + b \, w\right) + c \, (v - z) - 2 \, c \left(v - \frac{b}{a}\right) + D_4 \, z_{xx}. \label{secBru4}
        \end{align}
        We assume, again, that $a, b, c > 0$ and $D_1, D_2, D_3, D_4\geq 0$, but now $\eta$ may take either sign.
        
        This model has at most nine homogeneous steady states but, again, we focus only on the simplest one \eqref{eq:P_bruss}, the Jacobian matrix about which is the same as \eqref{eq:Brus_Jac}, except for a switch in the sign of the $(4,2)$ entry.
        
        Note now that the matrix 
        \begin{align*}
        	S_{24} = \begin{pmatrix}
        		- a^2 \, \eta - c & c
        		\\
        		-c & - a^2 \, \eta - c
        	\end{pmatrix}
        \end{align*}
        has the property that when $\tr\left(S_{24}\right) = 0$, $M_{24} = c^2 > 0$. Thus, this submatrix has a pair of complex conjugate eigenvalues that cross the imaginary axis when $\eta = -\frac{c}{a^2}$. Therefore, if the steady state is stable then, according to Theorem \ref{thm:bifurcations}, we can find positive diffusion rates such that the modified system goes through a wave instability when $\eta < -\frac{c}{a^2}$.
        
        In particular, choosing the parameter values
        \begin{align}
        	(a, b, c, \eta) = \left(0.9, 1.3, 4.8, -6.6\right), \label{eq:brus2_par}
        \end{align}
        we find the eigenvalues of $\jac \mbf f(\mbf P)$ to be approximately $-5.42567 \pm 1.1069 i$, and $-0.808327 \pm 8.62189 i$, which means that $\mbf P$ is a stable homogeneous steady state. Furthermore, we find that $S_{24}$ has eigenvalues that are approximately $0.546\, \pm 4.8 i$, which means that this matrix is unstable, while $(-1)^3 \, M_{124} = 151.552$ and $(-1)^3 \, M_{234} = 151.552$ are both positive. This implies that $(-1)^4 \, M_{1234} > 0$, $(-1)^3 \, M_{124} > 0$, $(-1)^3 \, M_{234} > 0$ and $(-1)^2 \, M_{24} > 0$. Therefore, we have an unstable submatrix of the Jacobian matrix of the system and, upon setting $D_2 = D_4 = 0$, there is a change of sign in the polynomial \eqref{eq:mupolynomial} as $\mu$ varies from 0 to $\infty$. Thus, Theorem \ref{biggestresult} shows the existence of wave instabilities for non-zero values of $D_2$ and $D_4$ in this case.
        
        Explicit parameter values for the wave bifurcation can be found if we choose $\left(D_1, D_2, D_3, D_4\right) = (10, \delta, 10, \delta)$. Figure \ref{fig:thirdexampleTH}(a) shows dispersion curves for different values of $\delta$ ranging from 0 to 0.025 in equal steps. The figure also shows that a wave bifurcation occurs at $\delta = 0.016939$.
        \begin{figure}
        	\centering
        	\includegraphics[width=0.45\textwidth]{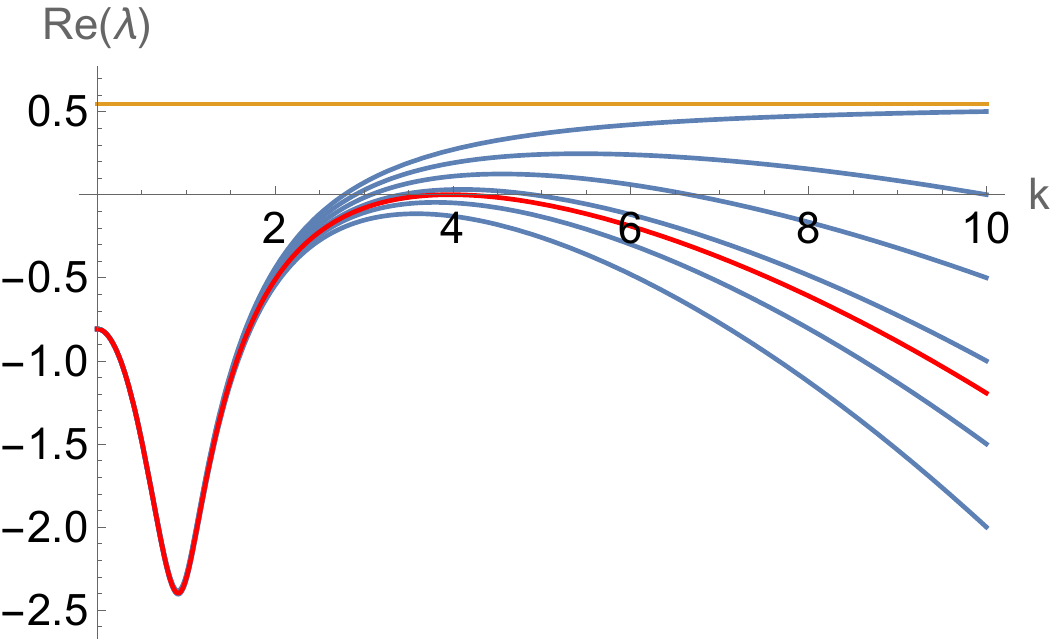}
        	\includegraphics[width=0.45\textwidth]{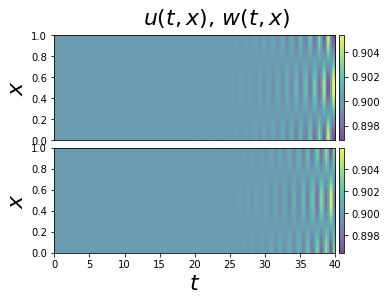}
        	\caption{(a) Similar to Fig.~\ref{fig:dispersionrelationTH}(a) but for system \eqref{Bru4} with the final equation replaced by \eqref{secBru4}, at parameter values \eqref{eq:brus2_par} and $\left(D_1, D_2, D_3, D_4\right) = (10, \delta, 10, \delta)$ as $\delta$ ranges from 0 to 0.025 in uniform steps. The distinguished (red) curve is for $\delta = 0.016939$, and the (orange) horizontal line is at $\Re(\lambda) = 0.546$. (b) Numerical integration of the system for $\delta = 0.005$, with initial conditions that represent a perturbation to $\mbf P$ of $10^{-5} \, \cos(2 \pi x)$ in each component.}
        	\label{fig:thirdexampleTH}
        \end{figure}
        
        In particular, when $\delta = 0.016939$, we find that the critical wavenumber is $k = 3.97209$, and the eigenvalues of the Jacobian matrix are, approximately, $-169.084$, $-159.468$, and  $\pm 4.80815 i$, confirming that the system indeed undergoes a wave bifurcation. Finally, Fig.~\ref{fig:thirdexampleTH}(b) shows the result of numerical simulation for a value of $\delta$ within the wave unstable parameter region, which indeed indicates the presence of a spatiotemporal instability.

\section{Discussion}\label{sec:discuss}
    This paper has established general conditions for the existence of both Turing and wave instabilities for reaction-diffusion systems with an arbitrary number of components, generalising and improving earlier partial results in the literature. Even though the fact that a Turing instability is a special case of a wave, there were no previous studies mixing both of these phenomena in systems with $n$ components and they need to be considered together in order to be able to find and characterize diffusion-driven instabilities you might find in this kind of system. Note that all our results were stated for reaction-diffusion systems with diagonal diffusion matrices. In principle, though, the results can readily be extended to the case that $\mathbb D$ is not diagonal, that is, where there are {\em cross-diffusion} terms.
    
    One simple way to extend the theory to the case that the off-diagonal diffusion coefficients are small is to first set them to zero. Then one can use the results developed in this paper to establish conditions for Turing and wave bifurcations. Next, one can use a homotopy parameter to increase the off-diagonal diffusion rates, and the theory will still hold provided no eigenvalues with positive real part cross the imaginary axis during the homotopy. 
    
    An alternative way to treat cases with cross-diffusion, assuming that $\mathbb D$ is positive semi-definite and has $n$ distinct real eigenvectors, could be to undertake a linear change of coordinates of the dependent variables to diagonalize $\mathbb D$. Further changes of parametrisation may be necessary to avoid diffusion constants appearing in the expressions for the local reaction terms $\mbf f$.
    
    At a more general level, everything we developed throughout the paper was done assuming that we are working with a reaction-diffusion system in $\mathbb R^m$. However, all of the theory can \evs{in principle} be generalized, \evs{up to a few technical assumptions} to the case of a system 
    \begin{align*}
    	\partial_t \mbf u &= \mbf f(\mbf u) + \mathbb D \, \mbf L(\mbf u),
    \end{align*}
    where $\mbf L$ is a linear self-adjoint operator in $\mbf u$ with negative eigenvalues, and boundary conditions accepting homogeneous steady states. In that case, the instabilities around those states we have described would spatially be described by eigenfunctions of that operator, rather than simple sinusoidal functions. An obvious example would be to consider $L(\mbf u)$ to be a constant matrix times the Laplacian, defined on a disk or sphere, see e.g.\cite{Anotida}. We shall leave the details of all such generalizations to future work.
    
    As a final comment, we note that the existence of a Turing bifurcation does not necessarily imply the onset of a stable periodic pattern. It seems in many reaction-diffusion systems the bifurcation may be subcritical, which means the periodic patterned states are unstable. Instead, in such systems, particularly on long domains, stable spatially localised patterns are typically observed; see e.g.~\cite{FahadWoods,MeronIssue} and references therein. The determination of whether a Turing (or wave) bifurcation is super- or sub-critical requires the computation of certain cubic coefficients that arise as part of the normal form. This can also help to understand the type of wave that persists after a wave bifurcation since, \evs{depending on boundary conditions, there may be} two types of waves that arise when a homogeneous steady state goes through that bifurcation, namely the \evs{{\em traveling wave}} and the \evs{{\em standing wave}} (see \cite{travellingorstanding,transition,wavelength}, and references therein). Such nonlinear analysis is beyond the scope of this paper, which has focused exclusively on linear calculations, but it will form the subject of future work by the present authors.

\bmhead{Acknowledgments}

The authors \evs{would} like to thank Andrew Krause of the University of Durham for helpful remarks.

E.~V-S. has received PhD funding from ANID, Beca Chile Doctorado en el extranjero, \evs{number} 72210071.

The authors have no conflicts of interest to declare.
%
%
%
%
%
%
%
%
%

\bibliography{sn-bibliography}

\end{document}